\documentclass[a4paper,12pt]{article}
\usepackage{verbatim}
\usepackage[utf8]{inputenc}
\usepackage{amssymb, amsmath}
\usepackage{amsthm}
\usepackage{amsbsy}
\usepackage{amsfonts}
\usepackage[colorlinks]{hyperref}
\usepackage{mathtools}
\usepackage{bookmark}
\usepackage{tikz}
\usepackage{xcolor}
\usepackage{float}
\usepackage[shortlabels]{enumitem}

\usepackage{dsfont}

\usepackage[sort,nocompress]{cite}

\usepackage{array,esint}

\newcommand{\R}{{\mathord{\mathbb R}}}
\newcommand{\Z}{{\mathord{\mathbb Z}}}
\newcommand{\N}{{\mathord{\mathbb N}}}
\newcommand{\C}{{\mathord{\mathbb C}}}

\newcommand{\HH}{\mathcal{H}}
\newcommand{\hh}{\mathfrak{h}}
\newcommand{\FF}{\mathcal{F}}

\newcommand{\UU}{\mathcal{U}}

\newcommand{\RR}{\mathcal{R}}

\newcommand{\TT}{\mathcal{T}}

\newcommand{\ran}{{\rm Ran}}

\def\one{\mathds{1}}

\def\inf{{\rm inf}\,}

\theoremstyle{plain}
\newtheorem{lemma}{Lemma}[section]
\newtheorem{theorem}[lemma]{Theorem}
\newtheorem{proposition}[lemma]{Proposition}

\newtheorem{corollary}[lemma]{Corollary}

\newtheorem{definition}[lemma]{Definition}

\newtheorem{remark}[lemma]{Remark}

\newtheorem{hyp}{\bf Hypothesis}
\newtheorem{hypr}{\bf Hypothesis}
\newtheorem{prope}{\bf Property}
\newtheorem{example}[lemma]{Example}
\newtheorem*{remark*}{Remark}

\numberwithin{equation}{section}

\newcommand{\degendim}{\textrm{d}}

\newcommand{\nn}{\nonumber}

\numberwithin{equation}{section}
\begin{document}

\title{Resonances  for  Atoms in  Dipole approximation of non-relativistic  QED}
\author{\vspace{5pt} David Hasler$^1$\footnote{
E-mail: david.hasler@uni-jena.de} and Markus Lange$^2$\footnote{E-mail: markus.lange@dlr.de}
\\
\vspace{-4pt} \small{$1.$ Department of Mathematics,
Friedrich Schiller  University  Jena} \\ \small{Jena, Germany }\\
\vspace{-4pt}
\small{$2.$ German Aerospace Center (DLR), Institute for AI-Safety and Security} \\
\small{Sankt Augustin \& Ulm, Germany}\\
}

\date{\today}
{
\maketitle

\begin{abstract}
We consider atoms or molecules coupled to the quantized electromagnetic radiation field in a dipole approximation. 
We show the existence of ground states and  resonance states in situations where the eigenvalues are degenerate and protected by a symmetry group.
We show that  ground states and  resonance states  as well as their  energies  depend analytically on the coupling constant.
Our results are an application of the result  in \cite{HasLan23-2}, which in turn is obtained using operator theoretic renormalization.
\end{abstract}

\tableofcontents

\section{Introduction}
\label{sec:SymSpinBoson}
Understanding the interaction between atoms and the electromagnetic field was one of the motivations
already present at the beginning of quantum mechanics. At low energies these
interactions can be understood in term of quantum mechanical models, where the electrons
of the atom are treated as non-relativistic quantum mechanical particles \cite{CohDiuLal86v1,CohDiuLal86v2}. Thereby the electromagnetic
field is quantized and described by a field of bosons with massless relativistic dispersion relation (called photons).
The interaction suppresses high energies by means of a so called ultraviolet
cutoff.   Such models are often referred to as non-relativistic qed and are defined in terms of a self-adjoint  operator, called Hamiltonian.
Whereas these models provide a good
description at low energies, they are not relativistically invariant. In contrast to relativistic quantum field theory which is plagued by infinities,
models of non-relativistic qed are mathematically well defined.

The Hamiltonian of models of non-relativistic qed is bounded from below \cite{BacFroSig98-1}.
The question about existence of ground states, that is, in mathematical terms,  whether the infimum of
the spectrum of the Hamiltonian is an eigenvalue has attracted much attention \cite{Spo98,BacFroSig99,Ger00,GriLieLos01,LieLos03}. 
For example the existence
of a ground state is an important ingredient in the analysis of scattering theory of quantum field theories.
The existence of ground states has been proven in large generality covering  situations which are of physical interest \cite{HasHer11-1,HasHer11-2, AbdHas12, BalDecHan.2019}.
Furthermore, the dependence of these ground states on external parameters, such as the coupling constant, has attracted attention.
A question which arises for example  in the calculation of physical quantities.  It has been shown in certain
situations that ground states depend analytically on the coupling constant \cite{GriHas09}. 
However these results are limited
to situations where the atomic part of the Hamiltonian,  without any coupling to the quantized field,
has a non-degenerate ground state or where the degeneracy is lifted at second order \cite{HasLan18-1}.  
Whereas these results
cover important physical situations,  many situations  are not covered. For example, if the
degeneracy is caused by an underlying symmetry, such as rotation symmetry or time reversal symmetry.
The goal of this paper is to address  exactly such situations, which so far, have not been covered.

Next to ground states, spectral properties at higher energies of Hamiltonians of non-relativistic qed  are of interest. 
In particular, they are important for the mathematical description of the experimentally observed  processes  of  emission and absorption of photons on atoms.
Starting from the works \cite{BacFroSig98-1, BacFroSig98-2, BacFroSig99} the following picture has emerged. Generically, the eigenvalues of the atomic Hamiltonian  dissolve into the continuum upon coupling to the quantized electromagnetic field.
This effect  is caused by  the fact that  photons are  massless and therefore can have  arbitrarily small energies.
One way of  analyzing spectral properties at higher energies  is
to study analytic  extensions of dilation of the Hamiltonian. Eigenvalues of these analytic extensions
can be viewed as remnants  of the eigenvalues of the atomic Hamiltonian. Thereby the eigenvalues
of the analytic extensions emanate from the eigenvalues of the atomic Hamiltonian. We
refer to the corresponding  eigenvectors as resonance states.
Resonance states have been investigated intensively for such models. In particular their existence
has been shown using operator theoretic renormalization or  iterated perturbation theory \cite{BacFroSig98-2, BacBalPiz17}. However, these results are limited to
situations where one studies perturbations of  non-degenerate  eigenvalues of the atomic Hamiltonian. In many physical situations these assumptions
are unnatural,  as already noted  above regarding  ground states.

In a recent paper \cite{HasLan23-2}  by the present authors  a  general result about  existence of
 ground states as well as resonance states was  established also for   degenerate situations, provided
 the degeneracy is caused by a symmetry and  is protected by the  interaction.  Moreover, in that
paper analytic dependence as a function of  the coupling constant for ground states and resonance states was
established. We note that analyticity of resonances as a function of the coupling constant was also studied in \cite{BalFauFroSch15, BalDecHan.2019}.

In the present paper, we  apply the abstract result  in \cite{HasLan23-2} to  concrete models.
We consider  atoms as well as  more generally ions and molecules.
 Thereby, we assume that  the  nuclei  are   fixed and that the coupling to the quantized
electromagnetic field is treated in a dipole approximation \cite{GriHas09}.
First we assume that the ground state space of the atomic Hamiltonian is one dimensional
or is an irreducible subspace of a symmetry group commuting with the Hamiltonian.
Then the ground state energy of the Hamiltonian has a basis of eigenvectors, such that both, ground state energy
and eigenvectors
depend analytically
on the coupling constant.
We show an analogous result for resonance states.  We consider an eigenvalue of the atomic Hamiltonian,
whose eigenspace is one dimensional or is an irreducible subspace of a symmetry group of the Hamiltonian.  Then there exists an analytic extension of the dilated Hamiltonian,
such that this extension has an eigenvalue emanating from the eigenvalue of the atomic Hamiltonian.
This eigenvalue has a basis of eigenvectors. The eigenvalue as well as the eigenvectors are depending analytically on the coupling constant and on the parameter of  dilation.
We note that the dipole approximation is made to simplify the analysis. We believe that  the algebraic
aspects of the proof do not change if one considered the standard model of non-relativistic qed. However,
 for the analytic part  a  much more involved  renormalization group analysis \cite{BCFS, Sig09 }   would be required.

The paper is organized as follows. In Section \ref{sec:model} we introduce the model and state hypothesis needed
for the main results.
 In Section \ref{sec:nondeg} we state the results in case of non-degenerate eigenvalues. In particular, the result concerning the analytic dependence of resonance states on the coupling parameter is of interest for its own and new, even in the non-degenerate case.
Next we consider the main results of this paper, which concerns degenerate situations. In Section \ref{sec:rotsymm} we
consider  rotational invariant systems, and state the main
results in situations, were the degeneracy is due to rotation invariance.
 In Section  \ref{sec:timereversal} we
consider  systems where the degeneracy is  due to time reversal symmetry and state the main results
in that case. In Section \ref{sec:examples}  we study    examples for systems with rotation invariance
and time reversal symmetry and  apply   the main results  to these concrete  examples as corollaries.
In Section \ref{sec:proof}  we prove the main results stated in Sections  \ref{sec:rotsymm} and \ref{sec:timereversal}.
The proof is based on verifying the assumptions of the general result in \cite{HasLan23-2}. For
the convenience of the reader and to fix the notation, we recall the main results of  \cite{HasLan23-2} in Appendix
\ref{app:ResultsfromOtherPaper} and notions concerning symmetries
and anti-unitary maps from  \cite{HasLan23-2} in  Appendix \ref{app:symmetries}.

\section{The Model} \label{sec:model}
\subsection{Atomic Hilbertspace}
The atomic Hilbert space, which describes the configuration of $N$ indistinguishable particles obeying Fermi-Dirac statistics with spin $s \in \frac{1}{2}\N_0$, is given by
$$
\HH_{{\rm at}} := \{ \psi \in L^2((\R^{3} \times \Z_{2s+1})^N) : \psi(\underline{x}_{\pi(1)},...,\underline{x}_{\pi(N)}) = {{\rm sgn}(\pi)} \psi(\underline{x}_1,...,\underline{x}_N) , \pi \in \mathfrak{S}_N \} ,
$$
where $\mathfrak{S}_N$ denotes the group of permutations of $N$ elements, ${\rm sgn}$ denotes the signum of the permutation,
and $\underline{x}_j = (x_j,\sigma_j) \in \R^3 \times \Z_{2s+1}$ denotes the coordinates of the $j$-th particle with spin component $\sigma_j$.
Hence we can write the configuration space of $N$ electrons, i.e., the electron Hilbert space, as
\begin{align} \label{defofHel}
\HH_{\rm el} = P_{a}   L^2(\R^3 \times \Z_{2s+1})^{\otimes N}
\end{align}
where $P_{a}$ stands for the projection onto the antisymmetric subspace.
Let $\mathcal{D}_s$ denote the representation space of $SU(2)$ with dimension $2s + 1$. We denote by $\tau_1$, $\tau_2$ and $\tau_3$ the generators of the Lie algebra $su(2)$ in the representation $\mathcal{D}_s$.
Recall, that for the case  $\mathcal{D}_{1/2}$ these generators are simply the Pauli-Matrices and for the case  $\mathcal{D}_{0}$ they are zero.
We use the notation  $\tau = (\tau_1,\tau_2,\tau_3)$ and define
$$
(S_j)_a   =  \left( \bigotimes_{l=1}^{j-1} \one_{\mathcal{D}_s}  \right) \otimes \tau_a \otimes  \left( \bigotimes_{l=j+1}^{N} \one_{\mathcal{D}_s}   \right) , \quad a=1,2,3 .
$$
The total spin operator is then defined by
$$
S = \sum_{j=1}^N S_j  .
$$
Moreover, the operator of orbital momentum for the $j$-th particle  is defined by
$$
(L_j)_a =   \sum_{b,c=1}^3  \epsilon_{a,b,c}    x_{j,b} p_{j,c} , \quad a=1,2,3 ,
$$
where $x_{j,b}$ denotes the $b$-component of the position corresponding to the $j$-th electron and similar for the momentum component $p_{j,c}$.
Thus the  total orbital momentum is
$$
L = \sum_{j=1}^N L_j .
$$
On the electron subspace $\HH_{\rm el}$ we consider the operator
\begin{align} \label{electronicham}
H_{\rm el} \coloneq - \Delta + V_{\rm el} = - \Delta + V  + I_{\rm SB}  ,
\end{align}
where the potential $V \in L^2_{\rm loc}(\R^{3N})$ is infinitesimally bounded with respect to $-\Delta$.
The spin-orbit coupling $I_{\rm SB}$ (see for example \cite{Thaller1992} and references therein) is of the form
\begin{align} \label{spinorbit}
I_{\rm SB} := I_{\rm SB}(\nu) := \sum_{j=1}^N ( \nu_j(x_j)  \wedge p_j ) \cdot S_j \,.
\end{align}
Here $\nu_j : \R^3  \to \R^3$, $j=1,...,N$, denote functions, which are chosen such that  $I_{\rm SB}$ is infinitesimally operator bounded with respect to $-\Delta$. Note that we can always drop the spin orbit coupling by setting $\nu_j$ to be identically zero.

\begin{remark}
{\rm The spin-orbit coupling is a relativistic correction obtained from a non relativistic limit of the Dirac operator \cite{Thaller1992}.
It is of the form \eqref{spinorbit}, where $\nu_j(x_j) \sim \nabla_{x_j} \phi_j(x_j)$ and $\phi_j$ is up to a constant given in terms of the potential energy of a self-consistent field
seen by the $j$-th particle.
Typically, one chooses the field of the nucleus.
Thus if $\phi_j(x) = \frac{Q}{|x|}$ is the Coulomb potential of a point charge $Q$ at the origin, then $\nu_j(x) \simeq \frac{x}{|x|^3}$. This fits with the spin orbit coupling often seen in physics textbooks (e.g. \cite{Tha05}),
\begin{align*}
	\frac{1}{2m^2 c^2} S \cdot L \frac{Z e^2}{r^3}  .
\end{align*}
We note that in this case the spin-orbit coupling is not bounded with respect to the Laplacian, which can be seen by the scaling argument in Remark~\ref{rem:scaling}.
However, if $\phi_j$ is from a nucleus the charge is smeared out and there is no problem with this term. }
\end{remark}

\begin{remark}\label{rem:scaling}
{\rm  We show that the spin-orbit term with $\nu(x) = \frac{x}{|x|}|x|^{-\alpha}$  is not Laplacian bounded if $\alpha > 1$.
Let $\varphi \in C_c^\infty(\R^3)$, vanishing in a neighborhood of the origin.
Then $\varphi_\lambda(x)  = \lambda^{3/2} \varphi(\lambda x)$ is a unitary transformation
and $\| \Delta \varphi_\lambda \| = \lambda^2 \|\Delta \varphi \|$ and on the other hand
 $\| \nu(x)  p \varphi_\lambda \| = \lambda^{1+\alpha} \| \nu(x)  p \varphi \|$.
}
\end{remark}

In addition to the properties above we also want that the spectrum of $H_{\rm el}$ has the structure
\begin{align} \label{specass}
\sigma(H_{\rm el}) \setminus  [\Sigma_{\rm el}, \infty )= \{ E_{{\rm el}, j} : j=0,..., M \}  ,
\end{align}
where
\begin{align*}
\Sigma_{\rm el} := \inf \sigma_{\rm ess}(H_{\rm el}) ,\quad M \in \N_0 \cup \{\infty\},
\end{align*}
and $(E_{{\rm el},j})_{j=0,...,M}$ is strictly monotone increasing satisfying $E_{{\rm el},j} < \Sigma_{\rm el}$ for all $j\in \{0,...,M\}$, see  Figure \ref{fig:SpectrumHel}.

\begin{figure}[h!]
\begin{center}
\begin{tikzpicture}[scale=1]
 \draw[-] (-4,0) -- (0.5,0) ;
  \draw[thick, dotted] (0.6,0) -- (0.9,0) ;
  \draw[->] (1,0) -- (10,0) coordinate (x axis);
\filldraw[red]  (-1,0) circle (2pt)   ;
\draw (-1,-0.1) -- (-1,0.1) node[above]{$E_0$};
\filldraw[red]  (0.3,0) circle (2pt)   ;
\draw (0.3,-0.1) -- (0.3,0.1) ;
\filldraw[red]  (1.8,0) circle (2pt)   ;
\draw (1.8,-0.1) -- (1.8,0.1) node[above]{$E_{j-1}$} ;
\filldraw[red]  (3.8,0) circle (2pt)   ;
\draw (3.8,-0.1) -- (3.8,0.1) node[above]{$E_j$};
\filldraw[red]  (5.8,0) circle (2pt)   ;
\draw (5.8,-0.1) -- (5.8,0.1)  node[above]{$E_{j+1}$}  ;
\filldraw[red]  (6.5,0) circle (2pt)   ;
\draw (6.5,-0.1) -- (6.5,0.1)   ;
\filldraw[red]  (7.3,0) circle (2pt)   ;
\draw (7.3,-0.1) -- (7.3,0.1)   ;
\filldraw[red]  (7.7,0) circle (2pt)   ;
\draw (7.7,-0.1) -- (7.7,0.1)   ;
\filldraw[red]  (7.9,0) circle (2pt)   ;
\draw (7.9,-0.1) -- (7.9,0.1)   ;
\draw (7.95,-0.1) -- (7.95,0.1)   ;
\filldraw[red]  (7.95,0) circle (2pt)   ;
\draw (7.975,-0.1) -- (7.975,0.1)   ;
\filldraw[red]  (7.975,0) circle (2pt)   ;
\draw (8,-0.1) -- (8,0.1) node[above]{$\Sigma_{\rm el}$};
\draw[very thick,red] (8,0) -- (10,0) ;
 \end{tikzpicture}
\end{center}
 \caption{\small  Graphical depiction of the spectrum for $H_{\rm el}$ if Hypothesis \ref{hypA}(iii) is satisfied. \label{fig:SpectrumHel}}
 \end{figure}
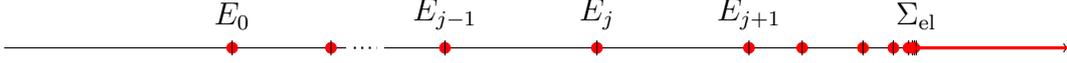

In total we therefore assume the following properties for the $N$-electron Hamiltonian $H_{\rm el}$.
\begin{hyp} \label{hypA}
The operator $H_{\rm el}$ has the following properties:
\begin{itemize}
\item[(i)] $V \in L^2_{\rm loc}(\R^{3N})$  is infinitesimally operator bounded with respect to $-\Delta$.
\item[(ii)] $I_{\rm SB}$   is infinitesimally operator bounded with respect to $-\Delta$.
\item[(iii)] $E_{\rm el} := \inf \sigma(H_{\rm el} )$ is an isolated eigenvalue of $H_{\rm el}$.
\end{itemize}
\end{hyp}

\begin{remark}{\rm
We note that Hypothesis~\ref{hypA} is, for example, satisfied for $\nu = 0$, i.e., no spin-orbit coupling, and
\begin{equation} \label{coulombpot}
V = V_c(x_1,....,x_N) := \sum_{j=1}^N  \frac{-Z}{|x_j|}  +  \sum_{i < j }^N  \frac{1}{|x_j - x_i|}  ,
\end{equation}
with $Z=N$. See for example  \cite{ReeSim2} for part (i) of Hypothesis~\ref{hypA} and \cite{Zhi60 ,ReeSim4} for part (iii).
}
\end{remark}

For two subspaces   $V$ and $W$  of a vector space, and $\|\cdot\|_W$ a norm on $W$ we write
$$
f \in V + [W]_\epsilon ,
$$
if for any $\epsilon > 0$ there exist $g \in V$ and $h \in W$ such that
$f = g  + h$  and    $\|h\|_W < \epsilon.$

\begin{lemma}\label{Spinbahn-1d}  Let $N\in \mathbb{N}$ and suppose that for $j=1,...,N$
$$
 \nu_j   \in  L^\infty(\R^3) + [L^3(\R^3)]_\epsilon     .
$$
Then  it follows that $I_{SB}$ is infinitesimally Laplacian bounded.
\end{lemma}

\begin{proof}
\underline{Step 1:} We start with the case $N=1$.
By the triangle inequality we have
\begin{align}
& \| I_{SB} \psi \| \leq  \sum_{\substack{ a, b=1:\\ a\neq b}}^3 2   \|  \nu_b  p_{a}  \psi \| .
\end{align}
We will use the  Gagliardo-Nierenberg-Sobolev-Inequality \cite[Theorem 1 in Section 5.6.1]{evans2010}
for  the following special case
\begin{align} \label{GNS}
\| u   \|_{L^6(\R^3)} \leq C_{\rm GNS} \| \nabla u \|_{L^2(\R^3)}  ,
\end{align}
 for all $ u \in C^1(\R^3) $.
Using H\"older inequality  and then  \eqref{GNS}  we find
for  $h \in  L^3(\R^3)$
\begin{align*}
\|  h p_{a}
 \psi \|_2^2  &  \leq  \|  | h|^2 \|_{3/2} \| | p_{a}  \psi|^2 \|_3    \leq  \|  h \|_{3}^2  \|    p_{a}  \psi  \|_{6}^2  \\
& \leq  \|   h\|_{3}^2    C_{\rm GNS}^2  \|   D  p_{a}  \psi  \|_{2}^2   \leq   \|   h | \|_{3}^2    C_{\rm GNS}^2  \|   \Delta  \psi  \|_2^2 .
\end{align*}
On the other hand for  $g \in L^\infty(\R^3)$ we find for any $\eta > 0$ that
\begin{align*}
\|  g p_{a}  \psi \|_2  &  \leq  \| g \|_\infty \| p_{a} \psi \|_2 \leq
 \| g \|_\infty  (\| \psi \|_2 \| p_{a}^2 \psi \|_2 )^{1/2} \leq \| g \|_\infty (
(4\eta)^{-1}  \|  \psi \|_2 + \eta  \| p_{a}^2 \psi \|_2 )
\end{align*}

Thus for $\nu_b = g + h$ with  $g \in L^\infty(\R^3)$ and $h \in  L^3(\R^3)$
we find from the above inequalities
\begin{align*}
&\| \nu_b p_{a} \psi \| =
\| (  g+ h) p_{a}  \psi \| \leq \|   g \psi \| + \| h p_{a}  \psi \| \\
& \leq  \|   h \|_{3}     C_{\rm GNS}\|   \Delta  \psi  \|_2 +  \| g \|_\infty (
(4\eta)^{-1}  \|  \psi \|_2 + \eta  \| p_{a}^2 \psi \|_2 )  \\
& \leq( \|   h \|_{3}     C_{\rm GNS}  +    \| g \|_\infty \eta ) \|   \Delta  \psi  \|_2
+  \| g \|_\infty
(4\eta)^{-1}  \|  \psi \|_2 .
\end{align*}
This shows the infinitesimally boundedness, by first choosing  $\|h \|_3$ sufficiently small
and then $\eta$.

\smallskip
\noindent
\underline{Step 2:} Now suppose $N \geq 2$.
By the triangle inequality
\begin{align} \label{sumISB}
& \| I_{SB} \psi \| \leq \sum_{j=1}^N \sum_{\substack{ a, b=1:\\ a\neq b}}^3 2  \|  \nu_b(x_j)   p_{j,a}  \psi \| .
\end{align}
Keeping the other variables fixed we find from Step 1  that for any $\epsilon > 0$
there exists a $C_\epsilon$ such that
\begin{align} \label{onepartestSB}
 \|  \nu_b(x_j)   p_{j,a}  \psi \|  \leq \epsilon \| -\Delta_j \psi \| + C_\epsilon \| \psi \| \leq
\epsilon \| \sum_{j=1}^N \Delta_j \psi \| + C_\epsilon \| \psi \| .
\end{align}
Since the sum in \eqref{sumISB}  is finite the claim now follows from   \eqref{onepartestSB}.
\end{proof}

We note that one can show a stronger result. For a proof see Section~\ref{proof:Spinbahn}.
\begin{lemma}\label{Spinbahn}
Let $\nu_j : \R^3 \to \R^3$  for  $j=1,...,N$ be  measurable functions such that
\begin{align}
& \forall \delta_0 > 0 ,  \quad  \sup_{|x| \geq \delta_0} |  \nu_{j}(x) | < \infty \label{refassump1}  \\
& \lim_{\delta \downarrow 0} \sup_{|x|  \in[0,\delta]}  |  |x|\nu_{j}(x)| = 0 .  \label{refassump2}
\end{align}
Then  $I_{SB}$ is infinitesimally Laplacian bounded.
Furthermore, assume in addition that $N=1$ and  $\nu_1 \in L^2(\R^3)$, then $I_{SB}$ is   relatively compact with respect to $-\Delta$.
\end{lemma}

\begin{remark} \label{Spinbahn-rem} {\rm
Let for some $a > 0$ and $b >2  $ and $C \geq 0$
\begin{equation} \label{UpperBoundNU}
| \nu_j(r) | \leq  \frac{ C r^{b}}{a+r^{b}} r^{-3}  , \quad r > 0 .
\end{equation}
Then the assumptions \eqref{refassump1}  and \eqref{refassump2} of the Lemma
 \ref{Spinbahn}. Thus by the Lemma  $I_{SB}$ is infinitesimally operator bounded with respect to $-\Delta$ (with the bound only depending on $C,b,a$).
 Furthermore $\nu_j \in L^2(\R^3)$ and so for $N=1$ the operator
 $I_{SB}$ is relatively compact with respect to $-\Delta$.
}
\end{remark}

\subsection{Fock space} \label{subsec:Fock}
Let $\hh = L^2(\R^3 \times \Z_2 )$ and let
\begin{align}\label{eq:FockSpace}
\mathcal{F} = \bigoplus_{n=0}^\infty \hh^{\otimes_s n}
\end{align}
be the Fock space modeling the quantized radiation field.
We denote the Fock vacuum by $\Omega$ and $a^*(\boldsymbol{k},\lambda)$ respective $a(\boldsymbol{k},\lambda)$ are the usual creation and annihilation operator satisfying canonical commutation relations, where $(\boldsymbol{k},\lambda) \in \R^3 \times \Z_2$.
For formal definitions of the annihilation and creation operator
and associated field operators we refer the reader to \cite{HasLan18-1,BacFroSig98-1}.
For $G \in L^2(\R^3\times \Z_2 ; \mathcal{L}(\HH_{\rm at}))$ we define
\begin{align} \label{eq:G-notation}
a(G) := \sum_{\lambda=1,2} \int G^*(\boldsymbol{k},\lambda) \otimes a(\boldsymbol{k},\lambda) d^3\textbf{k} , \quad    a^*(G) :=  \sum_{\lambda=1,2} \int G(\boldsymbol{k},\lambda)  \otimes a^*(\boldsymbol{k},\lambda) d^3\textbf{k}
\end{align}
 which are  densely defined closed linear operators in the Hilbert space.
We define the free field operator with dispersion relation $\omega$ by
\begin{equation}  \label{eq:fieldenergy}
H_f := \sum_{\lambda=1,2}\int{ \omega(\boldsymbol{k}) a^*(\boldsymbol{k},\lambda) a(\boldsymbol{k},\lambda) d^3\textbf{k}} ,
\end{equation}
which is defined in the sense of forms.

\subsection{The Hamiltonian}
A state of the system of an electron and transversal photons is described
by a vector in the Hilbert space
$$
\HH = \HH_{\rm el}\otimes \FF ,
$$
where we recall
$$
\HH_{\rm el} =   L^2(\R^{3};\mathcal{D}_s)^{\otimes_a N} \cong L_a^2((\R^3 \times \Z_{2 s + 1})^N)
$$
and $\mathcal{D}_s$ denotes the representation space of $SU(2)$ with $s \in \frac{1}{2} \N_0$
denoting the spin of the electrons.
For the applications we have in mind we shall consider the case of electrons $s = 1/2$ or $s=0$.
We choose units where $\hbar, c$ and four times the Rydberg energy are equal to one, and we express all positions in multiples of one half of the Bohr-radius, which in our units, agrees with the fine-structure constant $\alpha$.
For notational convenience, we set   $\kappa = \alpha^{1/2}$.
In these units the Hamiltonian of an atom with $N$ electrons in the dipole approximation of quantum electrodynamics reads \cite{Gri04,BacFroSig98-1,Spo04}
\begin{equation} \label{eq:defofhg}
H(\kappa) = H_{\rm el} +  \sum_{j=1}^N \left( \kappa^3 \chi(x_j) x_j \cdot E(0) + \kappa^5 S_j \cdot B(0) \right)  +   H_f .
\end{equation}

The interaction between the electrons and the quantized radiation field is given in terms of the following field operators evaluated at the origin
\begin{align*}
E(0) & =  \sum_{\lambda=1,2} \int_{} \rho(\boldsymbol{k}) \sqrt{|\boldsymbol{k}|}i \varepsilon(\boldsymbol{k},\lambda) \left( a^*(\boldsymbol{k},\lambda) - a(\boldsymbol{k},\lambda)  \right)  d^3 \boldsymbol{k}  , \\
B(0) & =  \sum_{\lambda=1,2} \int_{}  \rho(\boldsymbol{k}) \frac{1}{\sqrt{|\boldsymbol{k}|}} i  \boldsymbol{k} \wedge \varepsilon(\boldsymbol{k},\lambda)  \left( a(\boldsymbol{k},\lambda) - a^*(\boldsymbol{k},\lambda)  \right)  d^3\boldsymbol{k}  ,
\end{align*}
where
\begin{equation} \label{eq:poldefined}
\varepsilon :  \R^3 \setminus \{ 0 \}  \times \{1,2 \} \to \R^3
\end{equation}
 is a measurable function describing  the  so called  photon polarization vectors and satisfying for all $\lambda, \mu =1 ,2$ and $\boldsymbol{k} \in \R^3 \setminus \{ 0 \}$ the identities
\begin{equation}
\varepsilon(\boldsymbol{k},\lambda) \cdot \varepsilon(\boldsymbol{k},\mu) = \delta_{\lambda , \mu} , \quad
 \boldsymbol{k} \cdot \varepsilon(\boldsymbol{k},\lambda)  = 0 . 
\end{equation}
Moreover, for the result about resonances we shall assume in addition that
  $\varepsilon(\boldsymbol{k},\lambda) = \varepsilon(\boldsymbol{k}/|\boldsymbol{k}|,\lambda)$.
The function $\rho : \R^3 \to \R$ is a cutoff function.
 The function $\chi   \in C(\R^3)$ serves as a spacial-cutoff and is a function such that
\begin{equation} \label{eq:assumptionong}
\chi(x) = O\left(\frac{1}{|x|}\right) , \qquad   |x| \to \infty .
\end{equation}

\begin{hyp} \label{hyp:kappa}
The cutoff function $\rho$ is measurable and satisfies for some $\mu > 0$
\begin{equation} \label{ineqbound}
\int  (|\boldsymbol{k}|^2 + |\boldsymbol{k}|^{-1-2\mu} ) |\rho(\boldsymbol{k})|^2 d^3 \boldsymbol{k}  < \infty .
\end{equation}
\end{hyp}

\begin{remark} {\rm  \label{remoncouplibng1}
We note that  $\rho(\boldsymbol{k}) = e^{- (\boldsymbol{k}/\Lambda)^2}$, for any $\Lambda  > 0$,
 satisfies the assumptions of Hypothesis  \ref{hyp:kappa}, which is straight forward to see. In fact,  \eqref{ineqbound}  holds for any $\mu \in (0,1)$.   }
\end{remark}

Defining
\begin{align*} G_\kappa(\boldsymbol{k},\lambda) :=
 \frac{ \rho(\boldsymbol{k}) }{\omega(\boldsymbol{k})^{1/2}   } \sum_{j=1}^N  \left(  \kappa^3  \chi(x_j) x_j \cdot  |\boldsymbol{k}| i \varepsilon(\boldsymbol{k},\lambda)  +  \kappa^5 S_j \cdot
 i \boldsymbol{k} \wedge \varepsilon(\boldsymbol{k},\lambda)  \right)
\end{align*}
we can write the Hamiltonian as
\begin{align*}
	H(\kappa)   = H_{\rm el} +    W(\kappa)           +   H_f  ,
\end{align*}
where $W(\kappa) := a( G_{\overline{\kappa}}) + a^*(G_\kappa)$, see  \eqref{eq:G-notation}.

\subsection{Analytic dilation}\label{subsec:Analytic dilation}
Now let us consider analytic dilation, which are used to study resonances. We define for $\phi \in \hh$ the dilation operator
\begin{align*}
(u_{\rm ph}(\theta)\phi)(\boldsymbol{k},\lambda) = e^{-\frac{3\theta}{2} } \phi(e^{-\theta} \boldsymbol{k},\lambda ) , \quad
\end{align*}
where $(\boldsymbol{k},\lambda) \in \R^3 \times \Z_2$ and $ \theta \in \R$.
With this we define strongly continuous one-parameter unitary groups on
$\FF$ and $\HH$, respectively
\begin{align*}
U_{\rm ph}(\theta) = \Gamma(u_{\rm ph}(\theta)) ,  \quad
U(\theta) =  \one_{\HH_{\rm el}}  \otimes  U_{\rm ph}(\theta)  .
\end{align*}
Moreover for $\phi \in \hh$ we define the weighted norm
\begin{align*}
	\|  \phi\|_\mu = \left( \sum_{\lambda=1,2} \int \frac{|\phi(\boldsymbol{k},\lambda)|^2}{  |\boldsymbol{k}|^{2 + 2 \mu}  } d \boldsymbol{k}   \right)^{1/2}  .
\end{align*}

We can now define the main object of study.
For that let $\mathfrak{D}$ denote the domain of $-\Delta + H_f$ which is, by Hypothesis~\ref{hypA}, the same as the domain of $H_{\rm el} + H_f$.
We are interested in spectral properties of the mapping
\begin{align}
\R^2 \to \mathcal{B}(\mathfrak{D};\HH)\,;\,
 (\kappa, \theta) \mapsto  H(\kappa,\theta) :=  U(\theta) H(\kappa) U({\theta})^{-1}
= H_{\rm el}  +  W(\kappa,\theta)  +   e^{- \theta} H_f  , \label{eq:anaextofHg}
\end{align}
where
\begin{align} \label{eq:Wdipinfields}
W(\kappa,\theta) = &   \sum_{j=1}^N \left(\kappa^3  \chi(  x_j)   x_j \cdot E_\theta(0) + \kappa^5 S_j \cdot B_\theta(0) \right)
\end{align}
and
\begin{align*}
E_\theta(0) & =  e^{-\theta/2} \sum_{\lambda=1,2} \int_{} \rho(e^{-\theta} \boldsymbol{k}) \sqrt{|\boldsymbol{k}|}i \varepsilon(\boldsymbol{k},\lambda) \left( a^*(\boldsymbol{k},\lambda) - a(\boldsymbol{k},\lambda)  \right)  d^3 \boldsymbol{k}  , \\
B_\theta(0) & =   e^{-\theta/2} \sum_{\lambda=1,2} \int_{}  \rho(e^{-\theta} \boldsymbol{k}) \frac{1}{\sqrt{|\boldsymbol{k}|}} i  \boldsymbol{k} \wedge \varepsilon(\boldsymbol{k},\lambda)  \left( a(\boldsymbol{k},\lambda) - a^*(\boldsymbol{k},\lambda)  \right)  d^3\boldsymbol{k}  .
\end{align*}
Or equivalently we may write
\begin{align} \label{Wintdip}
W(\kappa,\theta) = a(G_{\overline{\kappa},\overline{\theta}}) +  a^*( G_{\kappa,\theta})
\end{align}
with
\begin{align*}
G_{\kappa,\theta}(\boldsymbol{k},\lambda) :=
 e^{-\theta/2} \frac{ \rho(e^{-\theta} \boldsymbol{k})}{\omega(\boldsymbol{k})^{1/2}} \sum_{j=1}^N
\left(  \kappa^3   \chi( x_j)  x_j \cdot  |\boldsymbol{k}| i \varepsilon(\boldsymbol{k},\lambda)  + \kappa^5 S_j \cdot
 i \boldsymbol{k} \wedge \varepsilon(\boldsymbol{k},\lambda)  \right)  .
\end{align*}

\begin{hyp} \label{hyp:analytph}  There exists a $\theta_b > 0$ such that  following analytic continuation properties hold.
The map $K_{\cdot} : \R \to L^2(\R^3\times \Z_2)$,  $\theta \mapsto   |\boldsymbol{k}|^{1/2} \rho(e^{-\theta} \boldsymbol{k})$  has an analytic continuation to a bounded function on $D_{\theta_b}$, such that
$\sup_{\theta \in D_{\theta_b}} \| K_{\theta} \|_\mu < \infty $.
\end{hyp}

\begin{remark} {\rm  \label{remoncouplibng2}
We note that $\rho(\boldsymbol{k}) = e^{- (\boldsymbol{k}/\Lambda)^2}$, for any $\Lambda  > 0$,
 satisfies the assumptions of Hypothesis~\ref{hyp:analytph}, which is straight forward to verify.  }
\end{remark}

\begin{remark} \label{rem:anacont} {\rm
We note that by Hypotheses  \ref{hyp:kappa} and   \ref{hyp:analytph}  the mapping \eqref{eq:anaextofHg} has a unique analytic continuation to $D_{\theta_b}$.
This analytic continuation is an analytic family of type (A).
This follows since the field operator \eqref{Wintdip} is $H_f$ bounded, see e.g. \cite{HasLan18-1,BacFroSig98-1}.
}
\end{remark}

\section{The non-degenerate case}\label{sec:nondeg}
The first results consider the case of a non-degenerate ground state (Theorem~\ref{gs:thmnondeg}) and of non-degenerate resonance states (Theorem~\ref{gs:thmnondegres}). We note that the first result is already well known \cite{GriHas09}. We state it in a form that fits with the later results considering more general degenerate situations.

We assume that the electronic Hamiltonian $H_{\rm el}$
has an eigenvalue $E_{\rm at}$ at the bottom of the spectrum with linearly independent eigenvectors $\varphi_{{\rm el},j}$, $j=1,...,d$. We will show the following property under various situations.

\begin{prope}(Ground State)   \label{prop:gs:thmnondeg-0}
There exists  a $\kappa_b > 0$ such that for all $\kappa \in D_{\kappa_b}$  the operator $H(\kappa)$  has $d$ linearly independent  eigenvectors $\psi_{\kappa,j}$, $j=1,...,d$,  with eigenvalue $E_{\kappa}$
such that on $D_{\kappa_b} $
\begin{itemize}
\item[(i)] $\kappa \mapsto \psi_{\kappa,j}$ is analytic.
\item[(ii)]   $\kappa \mapsto E_{\kappa}$ is analytic.
\item[(iii)] $E_{\kappa} = E_{{\rm el}} + O(\kappa^2)$,  $\psi_{\kappa,j} = \varphi_{{\rm el},j}  + O(g)$.
\item[(iv)]  $E_{\kappa} = \inf \sigma (  H_\kappa)$ if $\kappa \in D_{\kappa_b} \cap \R$.
\end{itemize}
\end{prope}

In the  non-degenerate case we recall the following result from \cite{GriHas09}.

\begin{theorem}(Ground State) \label{gs:thmnondeg-0}
Suppose Hypothesis \ref{hypA} and  \ref{hyp:kappa}   hold.
Suppose   that $E_{\rm el}$, the infimum
of the spectrum of $H_{\rm el}$ is  a non-degenerate eigenvalue   with eigenvector $\varphi_{{\rm el},1}$.
Then Property \ref{prop:gs:thmnondeg-0} holds.
\end{theorem}

Now let us consider situations, where we have additional dilation analyticity.
We will show the following property under various situations.

\begin{prope}(Dilated Ground State)    \label{prope:gs:thmnondeg}
There exists a $\theta_b > 0$ and a $\kappa_b > 0$ such that for all $\theta \in D_{\theta_b}$
and $\kappa \in D_{\kappa_b}$  the operator $H(\kappa,\theta)$  has  linearly independent  eigenvectors $\psi_{\kappa,\theta,j}$, $j=1,...,d$,  with eigenvalue $E_{\kappa,\theta}$
such that on $D_{\kappa_b} \times D_{\theta_b}$
\begin{itemize}
\item[(i)] $(\kappa,\theta) \mapsto \psi_{\kappa,\theta,j}$ is analytic.
\item[(ii)]   $(\kappa,\theta) \mapsto E_{\kappa,\theta}$ is analytic and does not depend on $\theta$.
\item[(iii)] $E_{\kappa,\theta} = E_{{\rm el}} + O(\kappa^2)$,  $\psi_{\kappa,\theta,j} = \varphi_{{\rm el},j}  + O(g)$.
\item[(iv)]  $E_{\kappa,\theta} = \inf \sigma (  H(\kappa,\theta))$ if $\kappa \in D_{\kappa_b} \cap \R$ and $\theta \in D_{\theta_b}$. 
\end{itemize}
\end{prope}

\begin{theorem}(Dilated Ground State) \label{gs:thmnondeg}
Suppose Hypothesis \ref{hypA}, \ref{hyp:kappa}, and    \ref{hyp:analytph}   hold.
Suppose   that $E_{\rm el}$, the infimum
of the spectrum of $H_{\rm el}$ is  a non-degenerate eigenvalue   with eigenvector $\varphi_{{\rm el},1}$.
Then  Properties   \ref{prop:gs:thmnondeg-0}  and \ref{prope:gs:thmnondeg}   hold.
\end{theorem}
\begin{proof}
	See Section~\ref{proof:thmnondeg}.
\end{proof}

Next we consider excited states of the electronic Hamiltonian $H_{\rm el}$. Whereas the
eigenvalues upon coupling to the radiation field typically dissolve into the continuum,
the eigenvalues of analytic dilatation of the Hamiltonian may persist.
We shall refer to the corresponding eigenvectors as \textit{resonance states}.
For this we assume that the electronic Hamiltonian $H_{\rm el}$ has an eigenvalue $E_{\rm at}$ below the essential spectrum  with linearly independent eigenvectors $\varphi_{{\rm el},j}$, $j=1,...,d$. We will show the following property under various situations.

\begin{prope}(Resonances) \label{prop:gs:thmnondegres}
For all $\vartheta_0 \in (0, \pi/2)$  there exists a $\theta_{\rm r} > 0$ and a $\kappa_{\rm b} > 0$ for all   $\theta \in D_{\theta_{\rm r}}(i \vartheta_0)$ and $\kappa \in D_{\kappa_{\rm b}}$  the operator $H(\kappa,\theta)$ has linearly independent eigenvectors $\psi_{j,\kappa,\theta}$, $j=1,...,d$, with eigenvalue $E_{\kappa,\theta}$ such that on $ D_{\kappa_{\rm b}} \times   D_{\theta_{\rm r}}(i \vartheta_0)$
\begin{itemize}
\item[(i)] $(\kappa,\theta) \mapsto \psi_{j,\kappa,\theta}$ is analytic.
\item[(ii)]   $(\kappa,\theta) \mapsto E_{\kappa,\theta}$ is analytic and does not depend on $\theta$.
\item[(iii)] $E_{\kappa,\theta} = E_{\rm el} + O(\kappa^2)$,  $\psi_{\kappa,\theta} = \varphi_{{\rm el},j} \otimes \Omega  + O(\kappa)$.
\end{itemize}
\end{prope}

\begin{theorem}(Resonances) \label{gs:thmnondegres}
Suppose Hypothesis \ref{hypA}, \ref{hyp:kappa}, and \ref{hyp:analytph} hold.
Suppose  $H_{\rm el}$ has  a non-degenerate eigenvalue $E_{\rm el}$ below the essential spectrum  with eigenvector $\varphi_{{\rm el},1}$. Then Property \ref{prop:gs:thmnondegres}   holds  with $d=1$.
\end{theorem}
\begin{proof}
	See Section~\ref{proof:thmnondegres}
\end{proof}

\section{Rotational Symmetry}\label{sec:rotsymm}
We introduce the  so called canonical  double covering  homomorphism
$$
\pi  : SU(2) \to SO(3)  ,  \quad U \mapsto \pi(U) ,
$$
where $\pi(U)_{l,j}$ is the unique element of $SO(3)$ such that
$$
U \sigma_j U^* =  \sum_{l=1}^3 \pi(U)_{l,j} \sigma_l  ,
$$
with  $\sigma_1, \sigma_2, \sigma_3$ denoting  the Pauli matrices.
On the one electron Hilbert space $L^2(\R^3;\mathcal{D}_s)$ we define
\begin{align} \label{defofsu2}
( \RR_{\rm el}(U) \psi)(x) = \mathcal{D}_s(U) \psi(\pi(U)^{-1} x) ,
\end{align}
where $\mathcal{D}_s$ denotes the  representation of $SU(2)$ with spin $s$.

By $F$ we shall denote the Fourier  transform and as usual by $F^{-1}$ its inverse.
We introduce the space of divergence free vector fields
$$
\mathfrak{v} := \{ v \in L^2(\R^3 ; \C^3) :   \sum_{j=1}^3 \boldsymbol{k}_j  (F v_j)(\boldsymbol{k}) = 0 \ , \ \text{a.e.} \ \boldsymbol{k}  \in \R^3  \} .
$$
Given a specific measurable choice for the polarization vectors we obtain a canonical identification with the one photon Hilbert space $\hh = L^2(\R^3 \times \Z_2)$.
This is the content of the following lemma.

\begin{lemma}\label{lem:canIdentif}  For  $\varepsilon$  as in  \eqref{eq:poldefined} the map
\begin{align*}
	\phi_\varepsilon : \hh \to \mathfrak{v} , \quad h \mapsto  \left( F^{-1} \sum_{\lambda=1,2} \varepsilon_j( \cdot , \lambda ) h(\cdot,\lambda) \right)_{j=1,2,3}  ,
\end{align*}
is unitary and its inverse is
$
(\phi_\varepsilon^{-1} v )(\boldsymbol{k},\lambda) = \varepsilon(\boldsymbol{k}, \lambda)  \cdot ( F v)(\boldsymbol{k})
$ for almost all $(\boldsymbol{k},\lambda) \in \R^3 \times \{1,2\}$.
\end{lemma}

\begin{proof}
The lemma  follows from a straight forward calculation using the properties of the polarization vectors.
Let $h \in \hh$.
Clearly, $\phi_\varepsilon$ is well defined, since $\boldsymbol{k} \cdot F(\phi_\varepsilon(h))(\boldsymbol{k}) = \boldsymbol{k} \cdot \sum_{\lambda=1,2}  \varepsilon(\boldsymbol{k},\lambda) h(\boldsymbol{k},\lambda) = 0$.
The map is an isometry, since
\begin{align*}
\| \phi_\varepsilon h \|^2  & =  \int d^3 \boldsymbol{k} \sum_{j=1}^3  \sum_{\lambda,\lambda'=1,2} \overline{ \varepsilon_j(\boldsymbol{k},\lambda) h(\boldsymbol{k},\lambda) }  \epsilon_j(\boldsymbol{k},\lambda') h(\boldsymbol{k},\lambda') \\
& = \int d^3 \boldsymbol{k}    \sum_{\lambda,\lambda'=1,2} \delta_{\lambda,\lambda'}  \overline{ h(\boldsymbol{k},\lambda) }   h(\boldsymbol{k},\lambda') = \| h \|^2 .
\end{align*}
Furthermore for $v \in \mathfrak{v}$  let $(\psi_\epsilon v)(\boldsymbol{k},\lambda) = \varepsilon(\boldsymbol{k},\lambda) \cdot (  Fv)(\boldsymbol{k})$.
Then
\begin{align*}
F (\phi_\varepsilon (\psi_\epsilon v)_j)(\boldsymbol{k})    & =  \sum_{\lambda=1,2} \varepsilon_j( \boldsymbol{k} , \lambda ) (\psi_\epsilon v)( \boldsymbol{k} ,\lambda)  \\
& =  \sum_{\lambda=1,2}  \varepsilon_j( \boldsymbol{k} , \lambda ) \sum_{l=1}^3  \varepsilon_l(\boldsymbol{k},\lambda) \cdot (  Fv_l)(\boldsymbol{k}) \\
& = Fv_j(\boldsymbol{k})  ,
\end{align*}
where we used that  $ \sum_{\lambda=1,2}  \varepsilon_j( \boldsymbol{k} , \lambda )   \varepsilon_l(\boldsymbol{k},\lambda)  = \delta_{j,l} - |\boldsymbol{k}|^{-2} \boldsymbol{k}_j \boldsymbol{k}_l$ and that $v$ is divergence free. This shows the surjectivity of $\phi_\varepsilon$ and that its inverse is given as claimed.
\end{proof}

Similar as for the one electron Hilbert space  we define for $v \in \mathfrak{v}$ the transformation
\begin{align*}
	( \RR_{\mathfrak{v}}(R) v)(x) = R v(R^{-1} x) ,
\end{align*}
where $R \in SO(3)$. Using the canonical identification from Lemma~\ref{lem:canIdentif} we define the mapping $\RR_\hh(R) = \phi_\varepsilon^{-1}  \RR_{\mathfrak{v}}(R)  \phi_\varepsilon$. Hence on $\hh$ we have the unitary transformation
\begin{align} \label{eq:defofrotrep}
( \RR_\hh(R) \varphi)(\boldsymbol{k},\lambda) & =  \sum_{\mu=1,2} (\varepsilon(\boldsymbol{k},\lambda)  \cdot R \varepsilon(R^{-1} \boldsymbol{k} , \mu) \varphi(R^{-1} \boldsymbol{k}, \mu)) ,
\end{align}
which depends on the choice of the polarization vectors. It is straight
forward to verify that $ \RR_\hh$ is a unitary representation of $SO(3)$ on
$\hh$. Thus for $U \in SU(2)$ we define the unitary mappings
\begin{align*}
 \RR_{\rm el}(U) & =  \RR_{\rm el}(U)^{\otimes N},  \\
  \RR_{\rm f}(U) & =    \Gamma( \RR_\hh(\pi(U))) , \\
 \RR(U) &=  \RR_{\rm el}(U)  \otimes  \RR_{\rm f}(U)
\end{align*}
on the Hilbert space $\HH_{\rm el}$, $\FF$ and
$\HH$. This defines a representation of $SU(2)$ on the Hilbert space $\HH$.

That this representation
has indeed the properties of rotation  symmetry can be seen in the following lemma,
which is straight forward to verify, for details we refer the reader for example to \cite{HasLan23-1}.

\begin{lemma}  \label{lem:propoftfrot}
The following holds for all $U \in SU(2)$.
\begin{itemize}
\item[(a)]  The operators  $\mathcal{R}_{\rm el}(U)$,   $\mathcal{R}_{\rm f}(U)$,
$\mathcal{R}(U)$
 are  unitary operators  on $\HH_{\rm el}$, $\FF$ and $\HH$, respectively.
\item[(b)]   $\mathcal{R}(U)  x_j    \mathcal{R}(U)^*     =  (\pi(U)^{-1} x)_j$, \,
$\mathcal{R}(U)  p_j    \mathcal{R}(U)^*     =  (\pi(U)^{-1} p)_j$, \, $\mathcal{R}(U) S_j \mathcal{R}(U)^* = \pi( U )^{-1} S_j $.
\item[(c)] $  \mathcal{R}(U)  H_{\rm f}   \mathcal{R}(U)^*   =  H_{\rm f}$.
\item[(d)]
If  $\rho$ is rotationally invariant, then
\begin{align*}
& \mathcal{R}(U)  E(0)   \mathcal{R}(U)^*     =  \pi(U)^{-1} E(0)  ,  
& & \mathcal{R}(U)  B(0)   \mathcal{R}(U)^*     =  \pi(U)^{-1} B(0) .  
\end{align*}
\item[(e)]  $\RR_{\rm f}(U)$ leaves the Fock vacuum as well as the one particle sector invariant.
\item[(f)]
For all $\tau\in\R$ we have $\mathcal{R}_{\rm f}(U)  U_{\rm ph}(\tau)= U_{\rm ph}(\tau) \mathcal{R}_{\rm f}(U) $.
\end{itemize}
\end{lemma}
\begin{proof}
 For (a) to (d) we refer the reader to \cite{HasLan23-1}. (e) and (f) follow directly from the definition.
\end{proof}

\begin{hyp} \label{hyp:rsym}   The following symmetry properties hold.
\begin{itemize}
\item[(i)] $V_{\rm el}$ commutes with $ \RR_{\rm el}(U)$ for all $U \in SU(2)$.
\item[(ii)]  The spacial cutoff $\chi$ is rotationally invariant.
\item[(iii)] The  photon momentum cutoff  $\rho$ is  rotationally invariant.
\end{itemize}
\end{hyp}

\begin{remark}{\rm
Clearly the Coulomb potential  $V(x_1,...,x_N) = \sum_{i > j} \big(  -\frac{Z}{|x_j|} + \frac{1}{|x_i-x_j|} \big)$  for any $Z \in \R$
satisfies  Hypothesis \ref{hyp:rsym} (i). }
\end{remark}

\begin{theorem}(Ground State) \label{gs:thmdegrot}
Suppose Hypothesis \ref{hypA},   \ref{hyp:kappa}, and \ref{hyp:rsym} hold.
Let  $E_{{\rm el}}$ denote the ground state energy  of     $H_{\rm el}$.
Suppose that the ground state space has dimension $d$ and   $\mathcal{R}_{\rm el} $  acts irreducible on that space. 
Then Property  \ref{prop:gs:thmnondeg-0} holds. Furthermore, if  in addition \ref{hyp:analytph} holds, then Property \ref{prope:gs:thmnondeg} holds.
\end{theorem}

\begin{proof}
	See Section~\ref{proof:gs:thmdegrot}
\end{proof}

Next we consider resonance states.

\begin{theorem}(Resonance  State) \label{gs:thmrotationres}
Suppose Hypothesis \ref{hypA},  \ref{hyp:kappa}, \ref{hyp:analytph} and \ref{hyp:rsym} hold.
Let  $E_{{\rm el}}$ be  an eigenvalue of     $H_{\rm el}$, below the essential spectrum.
Assume that its eigenspace has dimension $d$ and   $\mathcal{R}_{\rm el} $  acts irreducibly on that  eigenspace having  $\varphi_{{\rm el},j}$, $j=1,...,d$  as a basis of  eigenvectors. Then Property \ref{prop:gs:thmnondegres}   holds.
\end{theorem}
\begin{proof}
The proof follows in the same way as the proof of Theorem \ref{gs:thmnondegres} combined with the verification of Hypothesis \ref{HypII} (iii) as in the proof of Theorem \ref{gs:thmdegrot}. For the necessary notation we refer also to Table \ref{translationgesnondeg-2} in Section~\ref{proof:thmnondegres}.
\end{proof}

\section{Time Reversal Symmetry}\label{sec:timereversal}
Let $K$ denote complex conjugation on $L^2(\R^3; \mathcal{D}_s)$.
Define the operators
\begin{align*}
&
   \mathcal{T}_{{\rm p},s} := \left\{ \begin{array}{ll}  K     & , \quad \text{ if } \quad  s = 0  , \\
  (K   \sigma_2 )   &  , \quad \text{ if }  \quad s = 1/2   \end{array} \right.
 \end{align*}
and
\begin{align*}
&
   \mathcal{T}_{\rm el} :=  \bigotimes_{j=1}^N    \mathcal{T}_{{\rm p},s_j} .
\end{align*}
Let $\mathcal{K}_{\mathfrak{v}}$ denote complex conjugation in $\mathfrak{v}$,  and  let
\begin{equation} \label{eq:defofkh}
\mathcal{K}_{\hh} =    \phi_\varepsilon^{-1}  K_{\mathfrak{v}} \phi_\varepsilon
\end{equation}
denote its action on  $\hh$.
 Next we define operator of time reversal on the quantum field
\begin{equation} \label{eq:defoftf}
\mathcal{T}_{\rm f} := \Gamma(- \mathcal{K}_\hh) .
\end{equation}
We define the operator of time reversal in  the full Hilbert space by
\begin{equation} \label{eq:defoftimerev}
\mathcal{T} = \mathcal{T}_{\rm el} \otimes \mathcal{T}_{\rm f} .
\end{equation}

That this operator has the properties of the time reversal symmetry follows among other things from the following lemma.
For notations and definitions regarding anti-linearity and symmetry we refer the reader to Appendix \ref{app:symmetries}.

\begin{lemma}  \label{lem:propoftf}
The following holds.
\begin{itemize}
\item[(a)]  The maps  $\mathcal{T}_{\rm el}$, $\TT_{\rm f}$, and $\TT$ are anti-unitary operators.
\item[(b)] We have  $\TT_{\rm f}^2 = 1$, and
\begin{align*}
& \TT_{\rm el}^2 = (-1)^{\sum_{j=1}^N 2 s_j}  , \qquad \qquad\, \TT^2 =  (-1)^{\sum_{j=1}^N 2 s_j}  .
\end{align*}
\item[(c)] $\mathcal{T}  x_j  \mathcal{T}^* = x_j $, \quad $\mathcal{T} p_j  \mathcal{T}^* = -p_j $,  \quad $\mathcal{T}  S_j  \mathcal{T}^* = - S_j $.
\item[(d)] $ \mathcal{T}_{\rm f}  H_f   \mathcal{T}_{\rm f} ^*   =  H_f$.
\item[(e)]
If  $\overline{ \rho(\boldsymbol{k})  } =  \rho(-\boldsymbol{k})$,
then
\begin{align*}
& \mathcal{T}_{\rm f}  E(0)   \mathcal{T}_{\rm f}^*     =  E(0)  ,  
& \mathcal{T}_{\rm f}   B(0)  \mathcal{T}_{\rm f}^*     = - B(0) .  
\end{align*}
\item[(f)] $\mathcal{T}_{\rm f}$ leaves the Fock vacuum as well as the one particle sector invariant.
\item[(g)] For all $\tau\in\R$ we have $\mathcal{T}_{\rm f} U_{\rm ph}(\tau) = U_{\rm ph}(\tau)  \mathcal{T}_{\rm f}$.
\end{itemize}
\end{lemma}
\begin{proof}
 For (a) to (e) we refer the reader to \cite[Section 4.4]{HasLan23-1}. (f) and (g) follow directly from the definition.
\end{proof}

\begin{hyp}\label{hyp:tinvgs} We have that
\begin{itemize}
\item[(i)] $V_{\rm el}$ is symmetric with respect to  $\mathcal{T}_{\rm el}$,
\item[(ii)] $\rho(-\boldsymbol{k}) = \overline{\rho(\boldsymbol{k})}$.
\end{itemize}
\end{hyp}

\begin{remark} We note that (i) of the Hypothesis~\ref{hyp:tinvgs} is satisfied if $V_{\rm el}$ is a real valued potential.
\end{remark}

\begin{lemma} \label{lem:hypeimpliestimeinv}
We consider spin--$\frac{1}{2}$ particles, i.e.,    $\HH_{\rm el} = L^2(\R^3;\C^2)^{\otimes_a^N}$. If  Hypothesis  \ref{hyp:tinvgs} holds,  then $\mathcal{T}$ is a symmetry of $H_g$.
\end{lemma}
\begin{proof}
By Lemma~\ref{lem:propoftf} (a), we know that $\mathcal{T}$ is anti-unitary.
The property $\mathcal{T} H_g \mathcal{T}^* = H_g$ follows also from Lemma~\ref{lem:propoftf}, together with the fact that
\begin{align}\label{eq:KSKS}
	\mathcal{T}_{\rm el} S_j \mathcal{T}_{\rm el}^* = - S_j \,,
\end{align}
and that $\mathcal{T}_{\rm el}  ( - \Delta ) \mathcal{T}_{\rm el}^*  = -\Delta $. Eq.~\eqref{eq:KSKS} holds since $(K \sigma_2)\, \sigma_j\, (K \sigma_2)^* = - \sigma_j$ where $\sigma_j$ denotes the $j$-th  Pauli matrix.
\end{proof}

Let us recall the following Lemma known from Kramer's Degeneracy theorem \cite{LosMiySpo09, HasLan23-1}.
\begin{lemma} \label{lem:timeirred}  Suppose \ref{hyp:tinvgs} (i) holds. Let $E_{\rm el}$ denote an eigenvalue of $H_{\rm el}$. Assume $s=1/2$ and $N$ is odd.
 Then the multiplicity of  $E_{\rm el}$  is at least two. If the multiplicity is two, then $
\mathcal{T}_{\rm el}$ acts irreducibly on the eigenspace.
\end{lemma}
\begin{proof}
It follows from Lemma \ref{lem:propoftf} that  $\mathcal{T}_{\rm el}^2 = -1$.
Now by Hypothesis \ref{hyp:tinvgs} (i) we know that $\mathcal{T}_{\rm el} H_{\rm el}\mathcal{T}_{\rm el}^* =H_{\rm el}$.
Now let  $\psi$ be a normalized eigenvector $H_{\rm el}$ with eigenvalue $E_{\rm el}$. We claim that $\mathcal{T}_{\rm el} \psi$
is a normalized eigenvector of $H_{\rm el}$ with eigenvalue $E_{\rm el}$, and $\psi \perp \mathcal{T}_{\rm el} \psi$.
For this observe that $ H_{\rm el}\mathcal{T}_{\rm el} \psi  =\mathcal{T}_{\rm el} H_{\rm el} \psi  = E_{\rm el} \mathcal{T}_{\rm el}  \psi$.
On the other hand using antiunitarity $\langle \psi , \mathcal{T}_{\rm el} \psi \rangle = \langle\mathcal{T}_{\rm el}^{2} \psi ,  \mathcal{T}_{\rm el}\psi  \rangle
 = \langle  - \psi , \mathcal{T}_{\rm el} \psi \rangle = -  \langle \psi , \mathcal{T}_{\rm el} \psi \rangle $.
 And so $ \langle \psi , \mathcal{T}_{\rm el} \psi \rangle= 0$. By what we have shown,  any invariant subspace containing
 a nontrivial vector must have at least dimension two. This implies the irreduciblity.
\end{proof}

\begin{theorem}(Ground State)  \label{gs:thmdegtime}  Suppose Hypothesis \ref{hypA}, \ref{hyp:kappa},
and
\ref{hyp:tinvgs} hold. Furthermore assume $s=1/2$ and $N$ is odd.
Let  $E_{{\rm el}}$ denote the ground state energy  of     $H_{\rm el}$.
Suppose that the ground state space has dimension $d=2$.
Then  Property  \ref{prop:gs:thmnondeg-0} holds. Furthermore, if  in addition \ref{hyp:analytph} holds, then
 Property \ref{prope:gs:thmnondeg} holds.
\end{theorem}

\begin{proof}
	See Section~\ref{proof:thmdegtime}
\end{proof}

Next we consider resonance states.

\begin{theorem}(Resonance  State) \label{gs:thmtimeres} Suppose Hypothesis \ref{hypA}, \ref{hyp:kappa}, \ref{hyp:analytph}, and \ref{hyp:tinvgs} hold. 
Furthermore assume $s=1/2$ and $N$ is odd.
Let  $E_{{\rm el}}$ denote an eigenvalue of     $H_{\rm el}$ below the essential spectrum.
Suppose that the eigenspace has dimension $d=2$. Then Property \ref{prop:gs:thmnondegres}   holds.
\end{theorem}

\begin{proof}[Proof of  Theorem \ref{gs:thmtimeres}] 
The proof follows in the same way as the proof of Theorem \ref{gs:thmnondegres} combined with the verification of Hypothesis \ref{HypII} (iii) as in the proof of Theorem  \ref{gs:thmdegtime}. We also refer again to the notation in Table \ref{translationgesnondeg-2} in Section~\ref{proof:thmnondegres}.
\end{proof}

\section{Examples} \label{sec:examples}

Let us now discuss applications of the Theorems. 
In this section we assume that $s = \frac{1}{2}$. For this we recall the  particle Hilbert space in  \eqref{defofHel}
\begin{align} \label{defofHelSec6}
\HH_{\rm el} = P_{a}   L^2(\R^3 \times \Z_{2})^{\otimes N}  .
\end{align}
We study the atomic operator
\begin{align} \label{eq:Hel}
 H_{\rm el} =  \sum_{j=1}^N ( T(p_j)   +  V_1(x_j) )   + \sum_{\substack{ i,j  : \\  j < i }} V_2(x_i,x_j)   +   I_{\rm SB}(\nu).
\end{align}
where $V_1$ describes physically an external potential and $V_2$ a pair interaction potential.
$T : \R^3 \to (0,\infty)$ stands for the kinetic energy and has the following form
\begin{align*}
T(p) & = \frac{p^2}{2 m_{\rm el}}   , 
\end{align*}
where $m_{\rm el} >0$ is the electron mass.  We will choose units such that $m_{\rm el} = 1/2$.
 In this section we also assume that the spatial cutoff $\chi$, see \eqref{eq:assumptionong}, is rotationally invariant and
 \begin{align}\label{couplingfn}
 \rho(\boldsymbol{k}) = e^{- (\boldsymbol{k}/\Lambda)^2} ,  \quad \Lambda  > 0 .
 \end{align}
Note that \eqref{couplingfn}  implies that  Hypothesis~\ref{hyp:kappa} and \ref{hyp:analytph} hold, cf. Remarks \ref{remoncouplibng1} and \ref{remoncouplibng2}.
Furthermore, we assume that the spin orbit coupling is of the form
 \begin{align} \label{spinorbitexpl}
 I_{\rm SB}(\nu_\epsilon) = \sum_{j=1}^N (\nu_\epsilon(x_j) \wedge p_j ) \cdot S_j
 \end{align}
  with
\begin{align}\label{spinorbitexpl-2}
\nu_\epsilon(x)  :=  \min\{ \epsilon^{-3}, |x|^{-3} \}  ,  \quad  \epsilon > 0 .
\end{align}
We note that it follows from  Lemma \ref{Spinbahn-1d}, that $I_{\rm SB}(\nu_\epsilon)$  defined in \eqref{spinorbitexpl}
is infinitesimally bounded
 w.r.t. to $-\Delta$. Thus we can assume that    Hypothesis \ref{hypA} (ii)  holds.
Note that by   Kramers rule \cite{Kra30},  a folklore theorem, one can show that  for  $s=1/2$ and odd $N$ each eigenvalue of $H_{\rm el}$
has even multiplicity. Furthermore, one can show that the degeneracy persists if one
adds the quantized electromagnetic field, for details see  \cite{LosMiySpo09,HasLan23-1}. 

\subsection{Charged Particle with Spin }

In this subsection we consider  a single particle with spin 1/2, i.e.,   $N=1$ and $s=1/2$. Physically this corresponds to the situation of an electron.
First we consider the ground state. Here the situation follows directly from Theorem \ref{gs:thmdegrot}.

\begin{corollary} \label{gsrotinvex}  
Let  $H_{\rm el}=  - \Delta  - \frac{Z}{|x|}$ for some $Z > 0$.
Then the ground state energy  $E_{\rm el}$ of $H_{\rm el}$  is an isolated eigenvalue, which is twofold
degenerate, with linearly independent eigenvectors $\varphi_{{\rm el},j}$, $j=1,2$.
 Properties  \ref{prop:gs:thmnondeg-0} and \ref{prope:gs:thmnondeg}  hold with $d=2$.
\end{corollary}
\begin{proof} 
 For the proof we use  Theorem \ref{gs:thmdegrot}.
 Note that it is well known that the Coulomb potential is infinitesimally bounded with respect to $-\Delta$ and so Hypothesis \ref{hypA} (i)  holds.
 In addition \ref{hyp:rsym} holds by the assumptions.
 \end{proof}

 \begin{remark} 	{\rm  We note that alternatively Corollary  \ref{gsrotinvex}  could be proven by means of time reversal symmetry,
 see Corollary \ref{cor:time:ex:gs}.   }
 \end{remark}

Note,  that it is well known that the excited eigenvalues of    $-\Delta  - \frac{Z}{|x|}$ have  degeneracies which are
 not irreducible with respect to $SO(3)$, and  Theorem \ref{gs:thmrotationres}  is  not directly applicable.
 On the other hand, it is known that
relativistic corrections lift these degeneracies at least partially.
Thus let us add a  spin-orbit coupling, which is a relativistic correction.
For $\beta \in \R$ we will consider the operator
\begin{align}\label{spinorbitel}
H_{\rm H}(Z,\beta,\epsilon) := -\Delta - \frac{Z}{|x|} + \beta I_{\rm SB}(\nu_\epsilon)   .
\end{align}

 Let us first consider the ground state if the atomic part is given by \eqref{spinorbitel}.
  \begin{corollary} 
  Suppose  $H_{\rm el}= H_{\rm H}(Z,\beta,\epsilon) $ for some $Z > 0$ and $\epsilon > 0$.
Then for $|\beta| $ sufficiently small,  the ground state energy  $E_{\rm el}$ of $H_{\rm el}$  is an isolated eigenvalue, which is twofold degenerate, with linearly independent eigenvectors $\varphi_{{\rm el},j}$, $j=1,2$, and
 Properties   \ref{prop:gs:thmnondeg-0}  and  \ref{prope:gs:thmnondeg}
 hold for $d=2$.
\end{corollary}
 \begin{proof} 
 For the proof we will use Theorem \ref{gs:thmdegrot}.
To verify  Hypothesis \ref{hypA} we note that it follows from Lemma \ref{Spinbahn-1d}, that $I_{\rm SB}(\nu_\epsilon)$ is infinitesimally bounded
 w.r.t. to $-\Delta$. Since  $-\Delta - \frac{Z}{|x|}$ has a ground state with ground state space being an
 irreducible representation of $SU(2)$, it follows from perturbation theory, that for   $|\beta| $ sufficiently small
 this property remains preserved.
   In addition Hypothesis  \ref{hyp:rsym} holds by the assumptions.
 \end{proof}

 Next we consider excited states if the atomic part is given by \eqref{spinorbitel},
  which upon coupling turn in to  so called resonance states.

   \begin{corollary}  \label{excitedhydroen} 
   Suppose  $H_{\rm el}= H_{\rm H}(Z,\beta,\epsilon) $ for some $Z > 0$ and $\epsilon > 0$ and $\beta \in \R$.
   Suppose that   $E_{\rm el}$ is an   eigenvalue of  $H_{\rm el}$ below the essential spectrum.
   Assume that its eigenspace has dimension $d$ and  $\mathcal{R}_{\rm el} $  acts irreducibly on that  eigenspace having  $\varphi_{{\rm el},j}$, $j=1,...,d$  as a basis of  eigenvectors.
   Then
 Property \ref{prop:gs:thmnondegres} holds.
\end{corollary}
 \begin{proof} 
 The proof follows from Theorem \ref{gs:thmrotationres}.
 To verify  Hypothesis \ref{hypA} we note that it follows from Lemma \ref{Spinbahn-1d}, that $I_{\rm SB}(\nu_\epsilon)$ is infinitesimally bounded.
   In addition Hypothesis  \ref{hyp:rsym} holds by the assumptions.
 \end{proof}

Next we show that the irreducibility assumption of Corollary \ref{excitedhydroen}  is indeed satisfied for the lowest excited states of the hydrogen atom with spin-orbit coupling.

\begin{lemma} \label{lem:perturbationhyd}
Let $Z > 0$. Then there exist  $\epsilon_0 > 0$ and $\beta_0 > 0$ such that for all $\beta \in (-\beta_0,\beta_0) \setminus \{ 0 \}$
the operator $H_{\rm H}(Z,\beta,\epsilon_0)$ has eigenvalues $E_{1,1}(\beta)$, $E_{1,2}(\beta)$, $E_{1,3}(\beta)$ depending analytically on $\beta$ for $|\beta | < \beta_0$ and as functions of $\beta$ converge for $\beta \to 0$ to the first  excited eigenvalue of $-\Delta -  Z/|x| $.
Furthermore,  the values  $E_{1,1}(\beta)$ $E_{1,2}(\beta)$, $E_{1,3}(\beta)$ are different for nonzero $\beta$ with $|\beta| < \beta_0$.
The action $SU(2)$, defined in  \eqref{defofsu2},  acts irreducibly on the eigenspaces of $E_{1,1}(\beta)$, $E_{1,2}(\beta)$, $E_{1,3}(\beta)$.
\end{lemma}

\begin{remark} {\rm  Corollary  \ref{excitedhydroen} together with  Lemma \ref{lem:perturbationhyd}
imply the following.
Let $Z > 0$.
There exists  $\epsilon_0 >0$ and $\beta_0> 0$ such that for  all $\beta \in (-\beta_0 , \beta_0 ) \setminus \{ 0 \}$ the following holds.
If  $E_{\rm el}$ is the first, second, or third excited eigenvalue of $H_{\rm H}(Z,\beta,\epsilon_0)$ with eigenbasis   $\varphi_{{\rm el},j}$ $j=1,...,d$, then Property  \ref{prop:gs:thmnondegres} holds.
 }
\end{remark}

\begin{proof}[Proof of Lemma \ref{lem:perturbationhyd}]
This follows from first order perturbation theory. Let $\HH := L^2(\R^3;\C^2)$.
Note that $I_{SB}(\nu_\epsilon)$ is for any $\epsilon > 0$ infinitesimally $-\Delta$ bounded.
It is well known that the first excited eigenvalue  of $H_{\rm H}(Z,0,\epsilon_0)$,  denoted  $E_2$,
is degenerate. This eigenspace $\HH_2 := 1_{H_{\rm H}(Z,0,\epsilon_0) = E_2} \HH$ decomposes into   representations of $L$. Note that $H_{\rm H}(Z,0,\epsilon_0)$ commutes
with $L^2$ and $S^2$. It is well known that the eigenvalues of $L^2$ are given by $l(l+1)$ with $l \in \N_0$.
It is well known that with
\begin{align}
 \HH_{2,l} & := 1_{H_{\rm H}(Z,0,\epsilon_0) = E_2, L^2 = l(l+1)}\HH  ,
\end{align}
we have
\begin{align}
\HH_2 & = \bigoplus_{l=0,1} \HH_{2,l}  .
\end{align}
Now the summands on the right hand side still carry a ``spin degeneracy''. For the further decomposition
we define   $J = L+S$. Note that $J^2$ commutes with  $H_{\rm H}(Z,0,\epsilon_0)$ and $J^2$  has the eigenvalues given
by $j(j+1)$ with $j \in \frac{1}{2}\N$.  Again it is well known that with
\begin{align} \label{defeigen}
\HH_{2,l,j} :=  1_{H_{\rm H}(Z,0,\epsilon_0) = E_2, L^2 = l(l+1), J^2=j(j+1)} \HH .
\end{align}
we have
\begin{align}
\HH_2= \bigoplus_{(l,j) \in \{(0,1/2), (1,1/2 ), (1,3/2)\} } \HH_{2,l,j}  .
\end{align}
Now it is well known from group theory \cite{Sim.97} that the  summands  $\HH_{2,l,j}$   are irreducible with respect to $J$, i.e.,
 the representation $SU(2)$ defined in \eqref{defofsu2}.
Now we find from group theory, cf.  \cite[pages 217-218]{Schwabl2007},  that there exists a positive constant $c_{G} > 0$
 such that for  $\varphi \in \HH_{2,l,j} $
 \begin{align}
 L \cdot S \varphi  =  c_G
  \left\{ \begin{array}{ll} l , & j=l+1/2 \\ -l- 1 , & j=l-1/2 \end{array} \right\}  \varphi .
 \end{align}
By a calculation in \cite{Schwabl2007} there exists a constant $c_R > 0$
such that for  all normalized $\varphi \in \HH_{2,1,j} $,
 \begin{align}
  \langle \varphi , r^{-3}  \varphi \rangle =  c_R     .
 \end{align}
Thus  by rotation invariance of  $\nu_{\epsilon }$ and dominated convergence there exists a  $c_R(\epsilon)$
such that
 \begin{align}
 \langle \varphi , \nu_{\epsilon }  \varphi \rangle =  c_R(\epsilon)      , \quad \varphi \in \HH_{2,1,j} , j=1/2,3/2  ,
 \end{align}
 and $\lim_{\epsilon \downarrow 0} c_R(\epsilon) = c_R$.
Thus we find for normalized $\varphi \in  \HH_{2,l,j}$  that
\begin{align} \label{splitingfirsthyd}
 \langle  \varphi , I_{\rm SB}(\nu_\epsilon) \varphi \rangle
 = \left\{ \begin{array}{ll} 0  , & (l,j) = (0,1/2)  \\ - 2 c_G c_R(\epsilon) , & (l,j) = (1,1/2) \\
 c_G c_R(\epsilon ) , & (l,j) = (1,3/2) \end{array}  \right. .
\end{align}
Now observe that for $\epsilon$ sufficiently small the values on the right hand side of  \eqref{splitingfirsthyd} are distinct.
Thus the claim now follows from analytic perturbation theory \cite{ReeSim4}.
\end{proof}

\subsection{Atoms}
In this subsection we consider more general atoms or ions. Such quantum system have a natural $SU(2)$ symmetry.
For this we  recall the  particle Hilbert space in  \eqref{defofHelSec6}.

We study the atomic operator \eqref{electronicham} where we make the following explicit choice for the interaction
\begin{align} \label{atom}
H_{\rm at}(N,Z,\mu_{\rm e,e} , \beta, \epsilon) =    \sum_{j=1}^N   \left(   p_j^2 -   \frac{Z }{|x_j |} \right)  +   \sum_{\substack{ i,j  : \\  j < i }}  \frac{\mu_{\rm e,e}}{|x_i-x_j|}  + \beta  I_{\rm SB}(\nu_\epsilon)
\end{align}
for  $Z > 0$  the nuclear charge and  $\mu_{e,e} \geq 0 $,   a number which in physical applications is $\mu_{e,e} = e^2$.
We recall the definition given in Eqns. \eqref{spinorbitexpl} and \eqref{spinorbitexpl-2}.
Then Hypothesis \ref{hypA} (i) is satisfied, since the Coulomb potentials are infinitesimally bounded.  Furthermore, we already observed that
 $I_{SB}(\nu_\epsilon)$ is of the form  such that Hypothesis \ref{hypA} (ii) holds.

First we state the result about the ground state.

 \begin{corollary}\label{cor:atomion}
Suppose that $E_{\rm el}$  the ground state energy of  $H_{\rm el} = H_{\rm at}(N,Z,\mu_{\rm e,e} , \beta, \epsilon) $
is an  isolated eigenvalue with eigenspace being irreducible w.r.t. to $\mathcal{R}_{\rm el}$, having  $\varphi_{{\rm el},j}$, $j=1,...,d$  as a basis of  eigenvectors.
Then
Properties  \ref{prop:gs:thmnondeg-0}  and \ref{prope:gs:thmnondeg} hold.
\end{corollary}
\begin{proof}
The proof follows as a corollary  of Theorem
\ref{gs:thmdegrot}.
Hypothesis \ref{hypA} (iii), \ref{hyp:kappa} and \ref{hyp:rsym} hold by assumption.
\end{proof}

Next we consider resonances.

 \begin{corollary}\label{cor:atomionres}
Suppose that $E_{\rm el}$ is an eigenvalue  of  $H_{\rm el} = H_{\rm at}(N,Z,\mu_{\rm e,e} , \beta, \epsilon) $  below the essential spectrum.
If the   eigenspace is an  irreducible representation w.r.t. to $\mathcal{R}_{\rm el}$, having  $\varphi_{{\rm el},j}$, $j=1,...,d$  as a basis of  eigenvectors.
Then
Property \ref{prop:gs:thmnondegres}   holds. 
\end{corollary}
\begin{proof}
The proof follows as a corollary of
Theorem \ref{gs:thmrotationres}.
Hypothesis \ref{hypA} (iii), \ref{hyp:kappa}, \ref{hyp:analytph} and \ref{hyp:rsym} hold by assumption.
 \end{proof}

Let us now discuss when the  assumptions of
Corollary \ref{cor:atomion} and \ref{cor:atomionres}
are satisfied.

\begin{example} {\rm  The case $N =  2$. (Helium)
Then the Hamiltonian \eqref{atom} leaves the subspaces in the following direct sum decomposition
invariant
\begin{align}
P_{a} \left[ L^2(\R^3 ; \C^2 )^{\otimes 2 } \right] = [ L^2(\R^3)^{\otimes_s  2} \otimes (\C^2)^{\otimes_a 2}] \oplus
 [ L^2(\R^3)^{\otimes_a 2 } \otimes (\C^2)^{\otimes_s 2}] .
\end{align}
Elements  in the first and second summand are referred to as  Parahelium and  Orthohelium states, respectively.

\begin{itemize}
\item Suppose  $Z=2$ and $\mu_{\rm e,e} =1$. Then the Hamiltonian describes Helium.
In first approximation one neglects the spin orbit coupling and sets  $\beta = 0$.
In that case physical arguments indicate that the infimum of the spectrum, $E_{\rm el}$  of $H_{\rm at}(2,2,1 , 0, \cdot)$ is an isolated eigenvalue
with multiplicity one and the corresponding eigenvector $\varphi_{{\rm el},1}$ is a Parahelium state \cite{Schwabl2007}.
Now by analytic perturbation theory this property does not change for $|\beta|$ small.
Thus in such a situation the assumption of Theorem \ref{gs:thmnondeg} hold.
\item  Suppose  $Z=2$. For  sufficiently small  $|\mu_{\rm e,e}|$  and  $|\beta|$ one can
show by means of perturbation theory that the Hamiltonian has eigenvalues below the essential spectrum. Furthermore, the  eigenvalues have eigenspaces which may be  lifted because of  nonzero values for  $\mu_{\rm e,e}$
or $\beta$, such that the eigenspaces are  irreducible representation spaces  of $\mathcal{R}_{\rm el}$, see
\cite{Schwabl2007} where first order calculations have been performed, provided one chooses $\mu_{\rm e,e}$
or $\beta$ nonzero.
Thus in such a situations the assumption of  Corollary \ref{cor:atomionres}   hold.
\end{itemize}
}
\end{example}

For larger atoms we refer to the NIST Atomic Spectra Database \cite{NIST}. See also the electron shell model in the literature \cite{Schwabl2007}.

\subsection{Molecules and more general Potentials}
Let us now discuss applications of the Theorems to molecular systems. 
Recall the particle Hilbert space in \eqref{defofHelSec6}.
 We assume in this subsection that
\begin{align} \label{eq:Hgen}
 H_{\rm gen}(V) =  \sum_{j=1}^N  T(p_j)   + V(x_1,....x_N)
\end{align}
where $V$ is infinitesimally bounded w.r.t $-\Delta$ and symmetric w.r.t. interchange of particle coordinates.
We note that this includes molecules where the potential is given by the potentials of $M$ nuclei at sites $R_j \in \R^3$
with nuclear charges $Z_j > 0$
\begin{align} \label{defofmol}
V(\mu_{\rm e,e}; x_1,....,x_N) =  - \sum_{j=1}^M \frac{Z_j}{|x - R_j |} + \sum_{i,j: j < i }  \frac{\mu_{\rm e,e}}{|x_i-x_j|} ,
\end{align}
We note that it follows from \eqref{couplingfn} and the reality assumption of the potential that Hypothesis \ref{hyp:kappa}, \ref{hyp:analytph} and \ref{hyp:tinvgs} hold.
Moreover, the Hypothesis \ref{hypA} (i) is satisfied, since the Coulomb potentials are infinitesimally bounded.

Let us first consider the ground state.

\begin{corollary}\label{cor:time:ex:gs}  Suppose that the ground state energy $E_{\rm el}$ is an isolated eigenvalue of $H_{\rm el} =  H_{\rm gen}(V) $ with multiplicity $d=2$, having the $\varphi_{{\rm el},j}$, $j=1,2$ as eigenvectors. Then
Properties  \ref{prop:gs:thmnondeg-0}  and \ref{prope:gs:thmnondeg} hold.
\end{corollary}

\begin{proof}
	The proof follows as a corollary of Theorem \ref{gs:thmdegtime}. Note that Hypothesis \ref{hypA} (iii) holds by assumption.
\end{proof}

Now we consider resonances.

\begin{corollary} \label{cor:time:ex:res}  Suppose that  $E_{\rm el}$ is an  eigenvalue of $H_{\rm el} =  H_{\rm gen}(V) $ below the essential spectrum with multiplicity $d=2$, having the $\varphi_{{\rm el},j}$, $j=1,2$ as eigenvectors. Then
Property  \ref{prop:gs:thmnondegres} holds.
\end{corollary}

\begin{proof}
	The proof follows as a corollary of Theorem \ref{gs:thmtimeres}. Hypothesis \ref{hypA} (iii) again holds by assumption.
\end{proof}

Let us now consider cases where the assumption of Corollaries \ref{cor:time:ex:gs} and  \ref{cor:time:ex:res}
hold.

\begin{example} {\rm
 $N=1$.
\begin{itemize}
\item
Suppose $-\Delta + V$ has as an operator on $L^2(\R^{3})$ at the bottom
of the spectrum, $E_{\rm at}$, an isolated eigenvalue. Then the eigenvalue has a  positive eigenvector by  the Perron Frobenius theorem.
It follows  that $E_{\rm at}$ is a    non-degenerate eigenvalue of  $-\Delta + V$  as operator on $L^2(\R^{3})$.
Thus as an operator on $L^2(\R^3;\C^2)$ it is trivially two fold degenerate and remains at the bottom
of the spectrum and isolated.
The time reversal symmetry acts irreducibly on that subspace, see Lemma \ref{lem:timeirred}.  Then the assumptions of Corollary \ref{cor:time:ex:gs} holds.
\item Suppose $E_{\rm at}$ is a simple eigenvalue of  $-\Delta + V$  as operator on $L^2(\R^{3})$
below the essential spectrum.  Then it is trivially also
an  eigenvalue of $-\Delta + V$  as operator on $L^2(\R^{3};\C^2)$, which is two fold
degenerate and still below the essential spectrum. The time reversal symmetry acts again irreducibly by Lemma \ref{lem:timeirred} on that subspace.
Then the assumptions of Corollary \ref{cor:time:ex:res} hold.
\end{itemize} }
\end{example}

\begin{example} {\rm
Let  $N \geq 3$ be odd.
\begin{itemize}
\item
Suppose $ V$ is a sum of one body potentials $V_1$. Suppose the operator $(-\Delta + V_1)|_{L^2(\R^3)} $ has eigenvalues below the essential spectrum which we label in increasing order  $E_1 < E_2 < \ldots$ with multiplicities $n_1, n_2,...$.
Then we can fill the energy levels using the Pauli exclusion principle.
Thus suppose $N-1 = \sum_{j=1}^m 2 n_j$ for some $m$ and $E_{m+1}$
is non-degenerate. Then the ground state energy of the operator \eqref{eq:Hgen} on the space \eqref{defofHelSec6} is $\sum_{j=1}^m 2 n_j E_j + E_{m+1}  $ and it is two fold degenerate.
Hence the assumptions of Corollary \ref{cor:time:ex:gs} hold (Lemma \ref{lem:timeirred}).
Note by perturbation theory this property is preserved under small relatively bounded perturbations.

\item Similarly one can have excited eigenvalues. If $ E_{m+2}-E_{m+1}  < E_{m+1} - E_{m}$
and  $E_{m+2}$ is non-degenerate, then the first  excited state is obtained by putting the electron
in the $E_{m+1}$ state into the  $E_{m+2}$ state. The resulting state is an irreducible representation space.  If $ E_{m+2}-E_{m+1}  > E_{m+1} - E_{m}$ then one takes an electron out of the eigenspace $E_{m}$ and puts it into the space $E_{m+1}$.
This is the first excited eigenvalue and the corresponding eigenspace is irreducible.
Thus the assumptions of Corollary \ref{cor:time:ex:res} hold as above.

\end{itemize} }
\end{example}

\begin{example} {\rm
Let $N=2$.  Then we can decompose the Hilbert space as follows
 $$
 \HH_{\rm el} = \left(  L^2(\R^3 )^{ \otimes_s 2}  \otimes (\C^2)^{ \otimes_a 2} \right) \oplus
 \left(  L^2(\R^3 )^{ \otimes_a 2 }  \otimes( \C^2)^{ \otimes_s 2} \right).
 $$
These spaces are left invariant by the Hamiltonian.
\begin{itemize}
\item Now if the ground state $E_{\rm at}$ of the Hamiltonian is below the essential spectrum, then
it must be the ground state of the Hamiltonian acting on $L^2(\R^3)^{\otimes 2}$.
By Perron-Frobenius theorem the ground state must be positive hence it must be symmetric and is non-degenerate.
Thus the eigenspace of $E_{\rm at}$ of the  Hamiltonian acting
on the full Hilbert space  $\HH_{\rm el}$ including
the spin lies in  $L^2(\R^3 )^{ \otimes_s }  \otimes (\C^2)^{ \otimes_a }$.
Since   $(\C^2)^{ \otimes_a }$ is one dimensional is a  $E_{\rm at}$ non-degenerate eigenvalue of the Hamiltonian in $\HH_{\rm el}$.
Thus in such a situation the assumption of Theorem \ref{gs:thmnondeg} holds.
 \item Suppose $E_{\rm at}$ is a non-degenerate eigenvalue of the Hamiltonian restricted to  $L^2(\R^3)^{\otimes_s 2}$
but not of the Hamiltonian restricted to  $L^2(\R^3)^{\otimes_a 2}$. Then
 $E_{\rm at}$ is also a non-degenerate eigenvalue of the Hamiltonian acting on $\HH_{\rm el}$.
Thus in such a situation the assumption of Theorem \ref{gs:thmnondegres}   holds.
\item Now eigenvectors in   $ L^2(\R^3 )^{ \otimes_a 2}  \otimes( \C^2)^{ \otimes_s 2}$
 are at least   three fold degenerate.
Such a  subspace cannot be invariant by time reversal symmetry.
\end{itemize}
 }
\end{example}

\begin{example} {\rm
Let  $N \geq 4$ be even.
\begin{itemize}
\item
Suppose $ V$ is a sum of one body potentials $V_1$. Suppose the operator $(-\Delta + V_1)|_{L^2(\R^3)} $  eigenvalues
below the essential spectrum which we label in increasing order  $E_1 < E_2 < \ldots$ with multiplicities $n_1, n_2,\ldots$.
Then we can fill the energy levels using
the Pauli exclusion principle.  Thus if $N = \sum_{j=1}^m 2 n_j$ for some $m$.
 Then the ground state energy is $\sum_{j=1}^m 2 n_j E_j $ and it is non-degenerate.
Then the assumptions of  Theorem \ref{gs:thmnondeg}   hold.
Now by perturbation theory this property is preserved for small relatively bounded perturbations.

\item Now depending on the structure of the one particle Hamiltonian non-degenerate excited states of the total system may occur. Which can be seen by a similar argument as for the ground state case.
\end{itemize} }
\end{example}

\section{Proofs}\label{sec:proof}
In this section we provide the proofs that have been left out in order to make the main text more fluidly readable. We first prove certain estimates on the atomic Hamiltonian $H_{\textrm{el}}$ that are later used to verifying a Hypothesis from \cite{HasLan23-2} which is necessary for proving the main theorems. These are proved afterwards. We finish this section with the proof of Lemma~\ref{Spinbahn}.
\subsection{Estimates on the atomic Hamiltonian}
We denote the distance between an eigenvalue $E_{\rm el, j}$ of $H_{\rm el}$ and the rest of the spectrum by $\delta_j$, i.e.,
\begin{equation} \label{eq:GapAssumption}
\delta_j :=
{\rm dist}  \left( \sigma( H_{\rm el})  \setminus \{ E_{{\rm el},j} \} ,  E_{{\rm el},j} \right)  > 0
\end{equation}
For $\check{\delta}_j > 0$ we define the rescaled Hamiltonian
\begin{align} \label{hamrescaled}
 \check{H}_{{\rm el},j}(\theta) =  e^{\theta} \check{\delta}_j^{-1}  ( H_{\rm el} - E_{{\rm el},j} ) .
\end{align}
Note that the eigenvectors of the rescaled Hamiltonian are the same as for the original  Hamiltonian, i.e.,
\begin{align} \label{eqident22}
 \check{H}_{{\rm el},j}(\theta)   \ran P_{{\rm el},j}  =  0 .
\end{align}
and so  $\ran P_{{\rm el},j} $ consists of eigenvectors of $\check{H}_{{\rm el},j}(\theta)$ with
eigenvalue
\begin{align} \label{eqident22energy}
 \check{E}_{{\rm el}, j}(\theta)= 0 .
 \end{align}
Moreover, we set $\overline{{P}}_{{\rm el},j} :=\one_{\HH_{\rm el}} - P_{{\rm el},j}$.

The following proposition will be used to verify Hypothesis III from \cite{HasLan23-2} for the ground state case.

\begin{proposition}\label{veriatomIIIgs}
Let  $\theta_0 \in (0,\pi/2)$ and $\theta_1 \in (0,\infty)$. Then for any $\rho \in (0,\delta_0 / \check{\delta}_0)$ and
 $\mathcal{U} =    \{ \theta \in \C :  |  {\rm Im} \theta | < \theta_0 , \,  | {\rm Re} \theta |   < \theta_1  \}  \times D_{\rho}$ the following holds
$$
\sup_{(\theta,z)\in \UU}  \sup_{q \geq 0} \left\| \frac{q+1}{\check{H}_{{\rm el},0}(\theta) - e^{\theta} z + q} \overline{{P}}_{{\rm el},0} \right\|  < \infty .
$$
\end{proposition}

To illustrate the situation of   Proposition \ref{veriatomIIIgs}  and  Proposition  \ref{veriatomIIIex}, below, we refer to Figure    \ref{secondfigure}.

\begin{proof}  Let $d = \delta_0 / \check{\delta}_0$.
We have by the spectral theorem for any $q \geq 0$
\begin{align} \label{lem6.3bound}
& \left\| \frac{q+1}{\check{H}_{{\rm el},0}(\theta) - e^{\theta} z + q} \overline{{P}}_{{\rm el},0}  \right\|  \\
& = \sup_{s \in \sigma(H_{\rm el} ) \setminus \{ E_{{\rm el},0} \} }
\left|   \frac{q+1}{ e^\theta\check{\delta}_0^{-1} ( s  - E_{{\rm el},0})  - e^{\theta} z + q}  \right| \nn \\
& =  e^{- {\rm Re  }\theta}  \sup_{s \in \sigma(H_{\rm el} ) \setminus \{ E_{{\rm el},0} \} }
\left|   \frac{q+1}{ \check{\delta}_0^{-1} ( s  - E_{{\rm el},0})  -  z + e^{-\theta}q}  \right| \nn \\
& \leq   e^{- {\rm Re  }\theta}  \sup_{s \in \sigma(H_{\rm el} ) \setminus \{ E_{{\rm el},0} \} }
  \frac{q+1}{ |{ \rm Re } [\check{\delta}_0^{-1} ( s  - E_{{\rm el},0})  -   z + e^{-\theta}q]| }  \nn \\
  & \leq   e^{- {\rm Re  }\theta}
  \frac{q+1}{ | d  - {\rm Re}   z  +  q  {\rm Re}  e^{-\theta}| }\nn  \\
    & \leq   e^{- {\rm Re  }\theta}
  \frac{q+1}{ | d  - \rho  +  q  e^{- {\rm Re  }\theta}  \cos( - {\rm Im} \theta) |  }\nn \\
      & \leq   e^{- {\rm Re  }\theta}
  \frac{q+1}{  d  - \rho  +  q e^{- {\rm Re  }\theta}  \cos(  \theta_0)  } \nn \\
     & \leq
  \frac{e^{- {\rm Re  }\theta} }{  d  - \rho   } +   \frac{ 1  }{ \cos(\theta_0)  }  , \nn
\end{align}
where the last inequality follows by multiplying out the numerator.
\end{proof}

The following proposition is needed to verify Hypothesis III from \cite{HasLan23-2} for an excited state.

\begin{proposition}\label{veriatomIIIex}
Let  $\theta_0 \in (0,\pi/2)$ and $\theta_1 \in (0,\infty)$. Then for any $\rho \in (0,\delta_j / \check{\delta}_j )$ and
 $\mathcal{U} =    \{ (\theta , z ) \in \C^2  :  0 <   {\rm Im} \theta  < \theta_0 , \,  | {\rm Re} \theta |   < \theta_1  , |z| < \rho \sin({\rm Im} \theta) \} $ the following holds
$$
\sup_{(\theta,z)\in \UU}  \sup_{q \geq 0} \left\| \frac{q+1}{\check{H}_{{\rm el},j}(\theta) - e^{\theta} z + q} \overline{{P}}_{{\rm el},j}\right\|  < \infty .
$$
\end{proposition}

\begin{figure}
\begin{center}
 \begin{tikzpicture}[scale=1, rotate=22.5]
\filldraw[cyan!40!]  (8,0) -- (10,-1) -- (10.43,0) -- cycle;
\draw[very thick, red] (8,0) -- (10.43,0) ;
  \draw[thick, dotted] (0.6,0) -- (0.9,0) ;
  \draw[dotted] (1,0) -- (10,0) coordinate (x axis);
\filldraw[red]  (1.8,0) circle (2pt)   ;
\draw (1.8,-0.1) -- (1.8,0.1) node[above]{$E_{{\rm el},j-1}$} ;
\filldraw[red]  (3.8,0) circle (2pt)   ;
\draw (3.8,-0.1) -- (3.8,0.1) node[above]{$E_{{\rm el},j}$}  ;
\coordinate (a) at (3.8,0) ;
\coordinate (b) at (4.8,-0.43) ;
\coordinate (c) at (5,0) ;
\begin{scope}
	\clip (a) -- (b) -- (c) ;
	\draw [fill=lightgray!30, draw=black, dashed] circle[at=(a), radius =1cm] ;
\end{scope}
\draw (5.15,-0.05) node[below]{\tiny{${\rm Im} \theta$}}  ;
\draw (3.8,0.1) -- (3.8,-0.1)  node[below]{$0$}  ;
\filldraw[red]  (5.8,0) circle (2pt)   ;
\draw (5.8,-0.1) -- (5.8,0.1)  node[above]{$E_{{\rm el},j+1}$}  ;
\filldraw[red]  (6.5,0) circle (2pt)   ;
\draw (6.5,-0.1) -- (6.5,0.1)   ;
\filldraw[red]  (7.3,0) circle (2pt)   ;
\draw (7.3,-0.1) -- (7.3,0.1)   ;
\filldraw[red]  (7.7,0) circle (2pt)   ;
\draw (7.7,-0.1) -- (7.7,0.1)   ;
\filldraw[red]  (7.9,0) circle (2pt)   ;
\draw (7.9,-0.1) -- (7.9,0.1)   ;
\draw (7.95,-0.1) -- (7.95,0.1)   ;
\filldraw[red]  (7.95,0) circle (2pt)   ;
\draw (7.975,-0.1) -- (7.975,0.1)   ;
\filldraw[red]  (7.975,0) circle (2pt)   ;
\draw (8,-0.1) -- (8,0.1) node[above]{$\Sigma_{{\rm el}}$};
 \draw[->]  (3.8,0) -- (9.8,-3*0.828427124) ;
 \draw[-]  (3.8,0) -- (0.8,1.5*0.828427124) ;
\draw[very thick,cyan] (8,0) -- (10,-0.828427124) ;
\draw[very thick,cyan] (1.8,0) -- (3.8,-0.828427124) ;
\draw[very thick,cyan]   (4,-0.828427124-0.1*0.828427124) -- (4.2,-0.828427124-0.2*0.828427124)  ;
\draw[very thick,cyan]   (4.4,-0.828427124-0.3*0.828427124)  --   (4.6,-0.828427124-0.4*0.828427124)    ;
\draw[very thick,cyan]   (4.8,-0.828427124-0.5*0.828427124)  --   (5,-0.828427124-0.6*0.828427124)    ;
\draw[very thick,cyan] (3.8,0) -- (5.8,-0.828427124) ;
\draw[very thick,cyan]   (6,-0.828427124-0.1*0.828427124) -- (6.2,-0.828427124-0.2*0.828427124)  ;
\draw[very thick,cyan]   (6.4,-0.828427124-0.3*0.828427124)  --   (6.6,-0.828427124-0.4*0.828427124)    ;
\draw[very thick,cyan]   (6.8,-0.828427124-0.5*0.828427124)  --   (7,-0.828427124-0.6*0.828427124)    ;
\draw[very thick,cyan] (5.8,0) -- (7.8,-0.828427124) ;
\draw[very thick,cyan]   (8,-0.828427124-0.1*0.828427124) -- (8.2,-0.828427124-0.2*0.828427124)  ;
\draw[very thick,cyan]   (8.4,-0.828427124-0.3*0.828427124)  --   (8.6,-0.828427124-0.4*0.828427124)    ;
\draw[very thick,cyan]   (8.8,-0.828427124-0.5*0.828427124)  --   (9,-0.828427124-0.6*0.828427124)    ;
\draw[very thick,cyan] (6.5,0) -- (8.5,-0.828427124) ;
\draw[very thick,cyan]   (8.7,-0.828427124-0.1*0.828427124) -- (8.9,-0.828427124-0.2*0.828427124)  ;
\draw[very thick,cyan]   (9.1,-0.828427124-0.3*0.828427124)  --   (9.3,-0.828427124-0.4*0.828427124)    ;
\draw[very thick,cyan]   (9.5,-0.828427124-0.5*0.828427124)  --   (9.7,-0.828427124-0.6*0.828427124)    ;
\draw[very thick,cyan] (7.3,0) -- (9.3,-0.828427124) ;
\draw[very thick,cyan]   (9.5,-0.828427124-0.1*0.828427124) -- (9.7,-0.828427124-0.2*0.828427124)  ;
\draw[very thick,cyan]   (9.9,-0.828427124-0.3*0.828427124)  --   (10.1,-0.828427124-0.4*0.828427124)    ;
\draw[very thick,cyan] (7.7,0) -- (9.7,-0.828427124) ;
\draw[very thick,cyan]   (9.9,-0.828427124-0.1*0.828427124) -- (10.1,-0.828427124-0.2*0.828427124)  ;
\draw[very thick,cyan] (7.9,0) -- (9.9,-0.828427124) ;
 \end{tikzpicture}
\end{center}
\caption{\small  The typical  spectrum of the operator $\check{H}_{{\rm el},j}(\theta)$ for $0 < {\rm Im } \theta < \pi/2$ is indicated in red.
The set   $\check{H}_{{\rm el},j}(\theta) + [0,\infty)$ is indicated with cyan,
where the dashes indicate extension to $+\infty$.  \label{secondfigure}}
 \end{figure}
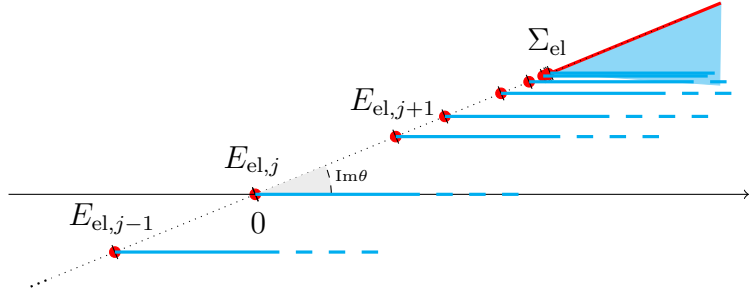

\begin{proof}   Let $d = \delta_j / \check{\delta}_j$ and $(\theta,z) \in \UU$.
We have by the spectral theorem for any $q \geq 0$
\begin{align} \label{eq:EstimateProp7.3-k?j}
& \left\| \frac{q+1}{\check{H}_{{\rm el},j}(\theta) - e^{\theta} z + q} \overline{{P}}_{{\rm el},j}  \right\|  \\
& \leq  \sup_{s \in \sigma(H_{\rm el} ) \cap(   E_{{\rm el},j} , \infty)  }
\left|   \frac{q+1}{ e^\theta \check{\delta}_j^{-1} ( s  - E_{{\rm el},j})  - e^{\theta} z + q}  \right| \nn \\
& +  \sup_{s \in \sigma(H_{\rm el} ) \cap( -\infty ,   E_{{\rm el},j} )  \} }
\left|   \frac{q+1}{ e^\theta \check{\delta}_j^{-1} ( s  - E_{{\rm el},j})  - e^{\theta} z + q}  \right| \nn.
\end{align}
Using that $|z| < \rho \sin({\rm Im} \theta)$ and $\rho < d$ the first term can be estimated as before in \eqref{lem6.3bound}
yielding
\begin{align} \label{eq:EstimateProp7.3-j<k}
& \sup_{s \in \sigma(H_{\rm el} ) \cap(   E_{{\rm el},j} , \infty)  }
\left|   \frac{q+1}{ e^\theta\check{\delta}_j^{-1} ( s  - E_{{\rm el},j})  - e^{\theta} z + q}  \right|
  \leq  \frac{e^{- {\rm Re  }\theta} }{  d  -  \rho \sin({\rm Im} \theta)  } +   \frac{ 1  }{ \cos(\theta_0)  }  .
\end{align}
To estimate the second term we use again
that $|z| < \rho \sin({\rm Im} \theta)$ and $\rho < d$ . More precisely, for $k < j$  and $q$ in a finite neighborhood of zero we will use the estimate
\begin{align*}
| e^{\theta} \check{\delta}_j  ^{-1}  ( E_{{\rm el},k}&- E_{{\rm el},j}) - e^\theta z + q | \\
& \geq  {\rm Im} \left(  e^{\theta} \check{\delta}_j  ^{-1}  ( E_{{\rm el},j} - E_{{\rm el},k}) + e^\theta z -  q \right)  \\
& \geq e^{{\rm Re} \theta} \left( \sin( {\rm Im} \theta) \check{\delta}_j^{-1}  ( E_{{\rm el},j}- E_{{\rm el},k}) - |z| \right)    \\
& = e^{{\rm Re} \theta}  \sin( {\rm Im} \theta) ( d - \rho )  .
\end{align*}
For $q \geq 2  | e^\theta  | (  \check{\delta}_j^{-1}   (E_{{\rm el},j}- E_{{\rm el},0}  ) +  \sin({\rm Im} \theta)  \rho)   =: q_1   $   we find
 \begin{align*}
 |e^{\theta} \check{\delta}_j^{-1}  ( E_{{\rm el},k}- E_{{\rm el},j}) - e^\theta z + q | & \geq  q -   | e^\theta |\check{\delta}_j^{-1} (E_{{\rm el},j}  - E_{{\rm el},0} ) +  \sin({\rm Im} \theta)  \rho)    \geq \frac{q}{2} .
 \end{align*}
 This gives
\begin{align} \label{eq:EstimateProp7.3-k<j}
& \sup_{s \in \sigma(H_{\rm el} ) \cap( -\infty ,   E_{{\rm el},j} )  \} }
\left|   \frac{q+1}{ e^\theta \check{\delta}_j^{-1} ( s  - E_{{\rm el},j})  - e^{\theta} z + q}  \right|
\nn \\  & \leq  \max  \left\{  \sup_{q \geq   q_1} \frac{2(q+1)}{q},    \sup_{q \leq  q_1  } \frac{q+1}{ |e^\theta|   \sin( {\rm Im} \theta)     (d -\rho)} \right\} \nn \\
  & \leq    \max  \left\{2 \left(1 + \frac{1}{q_1}\right) ,    \frac{1+q_1}{|e^\theta|  \sin( {\rm Im} \theta)   (d -\rho)}  \right\}
\end{align}
Inserting  \eqref{eq:EstimateProp7.3-k<j}  and \eqref{eq:EstimateProp7.3-j<k}
into \eqref{eq:EstimateProp7.3-k?j} yields the claim.
\end{proof}

\subsection{Proofs of the Main Theorems}
For the proofs of the theorems in the non-degenerate and degenerate cases we consider the following operator-valued function, which is an extension of $H(\kappa,\theta)$.
We define with $\underline{\kappa} = (\kappa_1,\kappa_2) \in \C^2$ , $g\in\R$
\begin{align}  \label{startingoperator}
H_{g}(\underline{\kappa},\theta)  & :=   H_{{\rm el}} + g   W(\underline{\kappa},\theta)  + e^{-\theta} H_f  .
\end{align}
Note that we have introduced the parameter $g$ here in order to have a similar notational form as in \cite{HasLan23-2}.
We define
\begin{align} \label{Wintdip-0}
W(\underline{\kappa},\theta) = a(G_{\overline{\underline{\kappa}},\overline{\theta}}) +  a^*( G_{\underline{\kappa},\theta})
\end{align}
with
\begin{align}\label{eq:GTheta} G_{\underline{\kappa},\theta}(\boldsymbol{k},\lambda) :=
 e^{-\theta/2} \frac{ \rho(e^{-\theta} \boldsymbol{k})}{\omega(\boldsymbol{k})^{1/2}} \sum_{j=1}^N
\left(  \kappa_1^3   \chi( x_j)  x_j \cdot  |\boldsymbol{k}| i \varepsilon(\boldsymbol{k},\lambda)  + \kappa_2^5 S_j \cdot
 i \boldsymbol{k} \wedge \varepsilon(\boldsymbol{k},\lambda)  \right)  .
\end{align}

\subsubsection{Proof of Theorem~\ref{gs:thmnondeg}}\label{proof:thmnondeg}

\begin{proof}[Proof of  Theorem \ref{gs:thmnondeg}]
We want to study the operator  defined in  \eqref{startingoperator} with $g \geq 0$ and $\underline{\kappa} \in \C^2$.

In order to obtain an operator  which is suitable for the renormalization analysis
we define the following auxiliary operator obtained by rescaling the original
operator such that the constant in front $H_f$ is one  and the new atomic part has the property
that the distance of the eigenvalue in question from the rest is suitably large.
Thus let   $\tau = - \ln \check{\delta}_0   \in \R$,
in order that the gap between the eigenvalue of interest and the rest of the spectrum  is one and the constant in front of $H_f$ is a one. We define
\begin{align}
\check{H}_g(\underline{\kappa},\theta)  & :=   e^{\theta} \check{\delta}_0^{-1}  ( U_{\rm ph}(\tau) H_{g}( \underline{\kappa}, \theta) U_{\rm ph}(\tau)^{-1}- E_{{\rm el},0} )
\end{align}
and write
\begin{align}  \label{mainverification}
\check{H}_g(\underline{\kappa},\theta)
& =  \check{H}_{{\rm el},0}(\theta) +  g   \check{W}(\underline{\kappa},\theta) + H_f ,
\end{align}
where we defined
\begin{align}
  \check{W}(\underline{\kappa}, \theta)  & :=  e^{\theta}\check{\delta}_0^{-1}U_{\rm ph}(\tau)   W(\underline{\kappa}, \theta) U_{\rm ph}(\tau)^{-1}   ,\label{defwtildefirst}
\end{align}
and recalled the notation
\begin{align}
 \check{H}_{{\rm el},0}(\theta) & :=  e^{\theta} \check{\delta}_0^{-1}  ( H_{\rm el} - E_{{\rm el},0} ). \label{defhtildefirst}   
\end{align}
In particular we find from  \eqref{Wintdip-0}  and \eqref{eq:GTheta} that
\begin{align} \label{Wintdip2}
 \check{W}(\underline{\kappa},\theta) =      a( \check{G}_{\overline{\underline{\kappa}},\overline{\theta}}) +
a^*(\check{G}_{\underline{\kappa},\theta})
\end{align}
with
\begin{align}\label{eq:GTheta2} \check{G}_{\underline{\kappa},\theta}(\boldsymbol{k},\lambda) :=
  e^{(\theta + \tau)/2}  \check{\delta}_0^{-1}\frac{  \rho(e^{-\theta - \tau} \boldsymbol{k}) }{\omega(\boldsymbol{k})^{1/2}} \sum_{j=1}^N
\left(  \kappa_1^3  \chi( x_j)  x_j \cdot  |\boldsymbol{k}| i \varepsilon(\boldsymbol{k},\lambda)  + \kappa_2^5 S_j \cdot
 i \boldsymbol{k} \wedge \varepsilon(\boldsymbol{k},\lambda)  \right)  .
\end{align}
In the following we will verify that Hypothesis \ref{HypI} - \ref{HypIV} in \cite{HasLan23-2} hold for the operator \eqref{mainverification} where we will use the notational translations given in Table~\ref{translationgesnondeg}.

\begin{table}[H] \caption{Notations}  \label{translationgesnondeg}
\begin{center}
\begin{tabular}{||c | c||}
 \hline
 \text{Notation here } & \text{Notation in  } \cite{HasLan23-2}  \\ [0.5ex]
 \hline\hline
 $(\underline{\kappa},\theta)$ & $s$  \\
 \hline
 $\check{H}_g(\underline{\kappa},\theta)$   & $H_g(s) $ \\
 \hline
 $\check{H}_{{\rm el},0}(\theta)$   & $H_{\rm at}(s)$  \\
 \hline
 $\check{G}_{\underline{\kappa},\theta}$  & $G_{1,s}$  \\
 \hline
 $\check{G}_{\underline{\kappa},\theta}$  & $G_{2,s}$  \\
 \hline
 $  \check{W}(\underline{\kappa},\theta) $ & $W(s) $ \\ [1ex]
 \hline
 $ \check{E}_{{\rm el},0} $ & $E_{\rm at}(s) $\\ [1ex]
 \hline
 $ P_{{\rm el},0}$  & $P_{\rm at}(s)$  \\ [1ex]
 \hline
\end{tabular}
\end{center}
\end{table}
Once the Hypothesis  \ref{HypI} - \ref{HypIV} are verified,  the claim follows directly from Theorem~2.10  in \cite{HasLan23-2}.
$$
$$

We will use Hypothesis \ref{hyp:analytph} to verify Hypothesis~\ref{HypI} in \cite{HasLan23-2}
for the function $(\underline{\kappa},\theta) \mapsto \check{G}_{\underline{\kappa},\theta}$ on the set $S \times D_{\theta_b}$ for  $S$  any bounded open subset of $\C^2$.
By   Hypothesis \ref{hyp:analytph} we know  that $(\boldsymbol{k},\lambda) \mapsto K_{\theta}(\boldsymbol{k},\lambda)  = |\boldsymbol{k}|^{1/2} \rho(e^{-\theta} \boldsymbol{k})$
is an $L^2(\R^3 \times \Z_2)$-valued analytic function in $\theta \in D_{\theta_b}$. By a straight forward calculation  $K_{\theta+ t} =  e^{2 t} u_{\rm ph}(t) K_{\theta}$ for  all $t, \theta \in \R$, which allows to extend  $K_{\theta}$   to an analytic function on the strip $S_{\theta_b} = \{\theta \in \C : | {\rm Im} \theta | < \theta_b \}$.

Let  $\partial_\theta K_\theta$ denote the complex derivative.
Define
\begin{align*}
J_{\underline{\kappa},\theta}(\boldsymbol{k},\lambda) :=
\frac{  \rho(e^{-\theta} \boldsymbol{k}) }{\omega(\boldsymbol{k})^{1/2}} \sum_{j=1}^N
\left(   \kappa_1^3  \chi( x_j)  x_j \cdot  |\boldsymbol{k}| i \varepsilon(\boldsymbol{k},\lambda)  + \kappa_2^5  S_j \cdot
 i \boldsymbol{k} \wedge \varepsilon(\boldsymbol{k},\lambda)  \right)  .
\end{align*}
Then  we find
\begin{align*}
& \left\| \frac{1}{h} \left(   J_{\underline{\kappa},\theta+h} -  J_{\underline{\kappa},\theta} \right) - \partial_\theta K_\theta \sum_{j=1}^N
\left(   \kappa_1^3  \chi( x_j)  x_j \cdot   i \varepsilon  + \kappa_2^5  S_j \cdot
 i  \frac{ \cdot }{|\cdot|}  \wedge \varepsilon  \right)  \right\|_{   L^2(\R^3\times \Z_2;\mathcal{L}(\HH_{\rm el})) }  \\
 & \leq \sup_{\boldsymbol{k}  ,\lambda} \left\|  \sum_{j=1}^N
\left(    \chi( x_j)  x_j  \cdot  i \varepsilon(\boldsymbol{k},\lambda)  + S_j \cdot
 i  \frac{ \boldsymbol{k}  }{|\boldsymbol{k}|}  \wedge \varepsilon(\boldsymbol{k},\lambda) \right)  \right\|_{\mathcal{L}(\HH_{\rm el})}  \\
 & \quad \times \left\| \frac{1}{h} ( K_{\theta + h} - K_\theta) - \partial_\theta K_\theta \right\|_{L^2(\R^3 \times \Z_2)} \\
 & \leq N ( |\kappa_1|^3 \sup_{x \in \R^3}   | \chi(x) x | + |\kappa_2|^5  )  \left\| \frac{1}{h} ( K_{\theta + h} - K_\theta) - \partial_\theta K_\theta \right\|_{L^2(\R^3 \times \Z_2)} \to 0  \quad ( h \to 0 ) .
\end{align*}
This shows that $\theta \to J_{\underline{\kappa},\theta}$ is an  $L^2(\R^3\times \Z_2;\mathcal{L}(\HH_{\rm el}))$-valued  analytic function
on $S_{\theta_b}$ for any $\underline{\kappa} \in \C^2$. Since  continuous products of analytic functions are analytic \cite{Osgood1899, Die69}, it follows that
$(\underline{\kappa}, \theta) \mapsto \check{G}_{\underline{\kappa},\theta} =   e^{ (\theta +\tau)/2} \check{\delta}_0^{-1}  J_{\underline{\kappa}, \theta+\tau}$ is an
 $L^2(\R^3\times \Z_2;\mathcal{L}(\HH_{\rm el}))$-valued  analytic function on $\C^2 \times S_{\theta_b}$.
To see the  boundedness   an estimate as above shows
 \begin{align*}
& \| J_{\underline{\kappa},\theta}  \|_{L^2(\R^3 \times \Z_2;\mathcal{L}(\HH_{\rm el}) )} \\
 & \leq \| K_\theta \|_{L^2(\R^3 \times \Z_2)} \sup_{\boldsymbol{k}  ,\lambda} \left\|  \sum_{j=1}^N
\left(  \kappa_1^3  \chi( x_j)  x_j  \cdot  i \varepsilon(\boldsymbol{k},\lambda)  + \kappa_2^5 S_j \cdot
 i  \frac{ \boldsymbol{k}  }{|\boldsymbol{k}|}  \wedge \varepsilon(\boldsymbol{k},\lambda) \right)  \right\|_{\mathcal{L}(\HH_{\rm el})}  \\
 &\leq N ( |\kappa_1|^3 \sup_{x \in \R^3} | \chi(x) x | +   |\kappa_2|^5 )   \| K_{\theta}  \|_{L^2(\R^3 \times \Z_2)}  .
 \end{align*}
 Now using
the  identity $ \| K_{\theta+\tau}  \|_{L^2(\R^3 \times \Z_2)}  =  e^{-2\tau} \| K_{\theta}  \|_{L^2(\R^3 \times \Z_2)} $
 and boundedness assumption of $K_\theta$ in   Hypothesis \ref{hyp:analytph}  we see that  $\check{G}_{\underline{\kappa},\theta}$ is bounded on $B \times D_{\theta_b}$
 for any bounded $B \subset \C^2$.
Let $\mu > 0$ be such that
\begin{align} \label{eqsupDDD}
\sup_{\theta \in D_{\theta_b}} \| K_{\theta} \|_\mu < \infty . \end{align}
Then using   \eqref{eq:GTheta2} and the translation Table \ref{translationgesnondeg} we find for any bounded subset $B \subset \C^2$
\begin{align*}
	& \sup_{s \in B \times D_{\theta_b}} \| G_{j, s} \|_\mu  \\
	& =\sup_{(\underline{\kappa},\theta) \in B\times D_{\theta_b}} \|\check{G}_{\underline{\kappa},\theta}\|_\mu  \\
	&=  \sup_{(\underline{\kappa},\theta) \in B \times D_{\theta_b}}    \left( \sum_{\lambda=1,2}\int \frac{1}{|\boldsymbol{k}|^{2+2\mu}}
\|\check{G}_{\underline{\kappa},\theta}(\boldsymbol{k} , \lambda)\|^2  	 d \boldsymbol{k}    \right)^{1/2}  \\
	&\leq \sup_{(\underline{\kappa},\theta) \in B \times (\tau + D_{\theta_b})} \\&\qquad   \left(\sum_{\lambda=1,2}\int \frac{ | e^{\theta} | |\rho(e^{-\theta} \boldsymbol{k})|^2}{\check{\delta}_0^{2} |\boldsymbol{k}|^{3+2\mu}} \left\|\sum_{j=1}^N \left( \kappa_1^3 \chi( x_j)  x_j \cdot  |\boldsymbol{k}| i \varepsilon(\boldsymbol{k},\lambda)  +  \kappa_2^5 S_j \cdot i \boldsymbol{k} \wedge \varepsilon(\boldsymbol{k},\lambda) \right)  \right\|^2  	 d \boldsymbol{k} \right)^{1/2}  \\
 	& \leq \frac{ N \sup_{\underline{\kappa} \in B} ( |\kappa_1|^3+|\kappa_2|^5)}{\check{\delta}_0} \left( \sup_{x \in \R^3} | \chi( x) x|  + 1\right) \sup_{\theta \in  D_{\theta_b} + \tau } |e^{\theta/2}| \left(\sum_{\lambda=1,2}\int \frac{|\rho(e^{-\theta} \boldsymbol{k})|^2  |\boldsymbol{k}|^2 }{|\boldsymbol{k}|^{3+2\mu}} d \boldsymbol{k}  \right)^{1/2} \\
	& \leq \frac{ N \sup_{\kappa \in B} ( |\kappa_1|^3+|\kappa_2|^5)}{\check{\delta_0}}  \left(\sup_{x \in \R^3} | \chi( x) x|  + 1\right) \sup_{\theta \in D_{\theta_b}} |e^{(\theta+\tau)/2}|\, \sup_{\theta \in D_{\theta_b}} \|  K_{\theta+\tau} \|_\mu \\
	&< \infty ,
\end{align*}
where we used  \eqref{eqsupDDD} together with the identity
\begin{align*}
  \| K_{\theta+\tau}  \|_{\mu}  =  e^{\tau-\mu \tau} \| K_{\theta}  \|_{\mu} ,
\end{align*}
and the boundedness of $B$ in the last step. Thus we have  verified Hypothesis~\ref{HypI}.

We observe that it follows immediately from the Definition  \eqref{defhtildefirst}   that Hypothesis~\ref{HypII} (i) in \cite{HasLan23-2}
is satisfied. Furthermore,  Hypothesis~\ref{HypII} (ii) follows directly  from the self-adjointness of $H_{{\rm el}}$.
Finally,  Hypothesis~\ref{HypII} (iii) follows trivially from  the non-degeneracy assumption.

Next, we want to verify Hypothesis~\ref{HypIII}.
By Proposition \ref{veriatomIIIgs}  for any
 $\theta_0 \in (0,\pi/2)$, $\theta_1 \in (0,\infty)$, and  $\rho \in (0,\delta_0 / \check{\delta}_0)$ we have for
 $U_0 :=    \{ \theta \in \C :  |  {\rm Im} \theta | < \theta_0 , \,  | {\rm Re} \theta |   < \theta_1  \} $
\begin{align}\label{eq:frompropvariatomIIgs:6}
\sup_{(\theta,z)\in U_0 \times D_\rho }  \sup_{q \geq 0} \left\| \frac{q+1}{\check{H}_{{\rm el},0}(\theta) - e^{\theta} z + q} \overline{{P}}_{{\rm el},0} \right\|  < \infty .
\end{align}
 For $\rho \in (0,1) $  assume
$
{\rm Re} \theta \geq - \ln 2 \rho  .
$
Then for any
$z \in D_{1/2}$ we  have
\begin{align} \label{eq:frompropvariatomIIgs:7}
|e^{-\theta} z | = e^{-{\rm Re} \theta} |z| \leq  e^{-{\rm Re} \theta} \frac{1}{2} \leq \rho .
\end{align}
Thus $ e^{-\theta} D_{1/2} \subset D_\rho $ and so  $D_{1/2} \subset e^{\theta} D_\rho $. We conclude from \eqref{eq:frompropvariatomIIgs:6} and \eqref{eq:frompropvariatomIIgs:7}  that for
\begin{align*}
U =  U_0 \cap \{ \theta \in \C : {\rm Re} \theta \geq - \ln 2 \rho \}
\end{align*}
we find
\begin{align} \label{IIIveri1-0}
\infty 	&  >  \sup_{(\theta,z) \in U \times D_{1/2} } \sup_{q \geq 0}
	\left\| \frac{q+1}{\check{H}_{{\rm el},0}(\theta)- z + q }
	\overline{P}_{{\rm el},0}(\theta)\right\| \\
	&= 	 \sup_{(s,z) \in U  \times D_{1/2} } \sup_{q \geq 0}
	\left\| \frac{q+1}{H_{{\rm at}}(s)- z + q }
	\overline{P}_{{\rm at}}(s)\right\|  \nonumber
\end{align}
where the last line follows from the notation introduced at the beginning of the proof.
Now  we see from \eqref{IIIveri1-0}
that Hypothesis  \ref{HypIII} follows.
We now choose   $\check{\delta}_0 = \delta_0 $.
Then $\delta_0 / \check{\delta}_0 = 1$. Thus   we can choose  $\rho =3/4$
and so $- \ln 2 \rho  < 0$.

Finally we observe that Hypothesis \ref{HypIV} holds for the  set  $X = \C \times \C$ as an immediate
consequence of   \eqref{defhtildefirst} and the notation in Table  \ref{translationgesnondeg}.
\end{proof}

\subsubsection{Proof of  Theorem \ref{gs:thmnondegres}}\label{proof:thmnondegres}
\begin{proof}[Proof of  Theorem \ref{gs:thmnondegres} ]
We will work with the operator  as defined  in \eqref{startingoperator}, i.e.,
\begin{align}  \label{startingoperator-2nd}
H_{g}(\underline{\kappa},\theta)  & :=   H_{{\rm el}} + g   W(\underline{\kappa},\theta)  + e^{-\theta} H_f  .
\end{align}

We will use  the definitions  \eqref{eq:GapAssumption}--\eqref{eqident22energy}, in particular we recall
\begin{equation} \label{eq:GapAssumption2}
\delta_j :=
{\rm dist}  \left( \sigma( H_{\rm el}(0))  \setminus \{ E_{{\rm el},j} \} ,  E_{{\rm el},j} \right)  > 0 .
\end{equation}
In order to obtain an operator  which is suitable for the renormalization analysis around the eigenvalue
we define similarly as for the ground state  the following rescaled operator.
 Let  $\tau = - \ln  \check{\delta}_j  \in \R$ and
\begin{align}
\check{H}_{j; g}(\underline{\kappa},\theta )  & :=   e^{\theta} \check{\delta}_j^{-1}  ( U_{\rm ph}(\tau) H_{g}(\underline{\kappa},\theta)   U_{\rm ph}(\tau)^{-1}- E_{{\rm el},j}(\theta) )   .
\end{align}
We can write
\begin{align} \label{mainoperatorworkwithresonances}
\check{H}_{j;g}(\underline{\kappa},\theta )
& =  \check{H}_{{\rm el},j}(\theta  ) +  g  \tilde{W}_j(\underline{\kappa} ,\theta) + H_f ,
\end{align}
where we defined
\begin{align}
  \check{W}_{j}(\underline{\kappa},\theta) & :=  e^{\theta} \check{\delta}_j^{-1}U_{\rm ph}(\tau)   W(\underline{\kappa},\theta) U_{\rm ph}(\tau)^{-1}   . \label{resHatchoice2}
\end{align}
Note that 
\begin{align*}
  \check{H}_{{\rm el},j}(\theta)  P_{{\rm el},j} =  e^\theta \check{\delta}_j^{-1} (  H_{{\rm el}}(\theta) - E_{{\rm el},j} )   P_{{\rm el},j}    = 0 .
\end{align*}
So $\check{E}_{{\rm el},j} := 0$ is an eigenvalue of $
 \check{H}_{{\rm el},j}(\theta)
$ with eigenprojection   $ P_{{\rm el},j}    $.
Furthermore,  we find from  \eqref{Wintdip-0}  and \eqref{eq:GTheta} that
\begin{align} \label{Wintdip2j}
\check{W}_{j}(\underline{\kappa},\theta) =      a( \check{G}_{j; \overline{\underline{\kappa}},\overline{\theta}}) +
a^*(\check{G}_{j;\underline{\kappa},\theta})
\end{align}
with
\begin{align}\label{eq:GTheta2j} \check{G}_{j; \underline{\kappa},\theta}(\boldsymbol{k},\lambda) & :=
 e^{(\theta + \tau)/2}  \check{\delta}_j^{-1}\frac{  \rho(e^{-\theta - \tau} \boldsymbol{k}) }{\omega(\boldsymbol{k})^{1/2}} \nonumber \\
& \quad \times  \sum_{l=1}^N
\left(   \kappa_1^3  \chi( x_l)  x_l \cdot  |\boldsymbol{k}| i \varepsilon(\boldsymbol{k},\lambda)  + \kappa_2^5  S_l \cdot
 i \boldsymbol{k} \wedge \varepsilon(\boldsymbol{k},\lambda)  \right)   .
\end{align}

Thus we will verify that Hypothesis \ref{HypI} - \ref{HypIII}  of  Theorem~2.10  in \cite{HasLan23-2} hold for the operator \eqref{mainoperatorworkwithresonances}
\begin{align}  
(\underline{\kappa}, \theta) \mapsto \check{H}_{j;g}(\underline{\kappa},\theta )
\end{align}
on the set $X :=  S  \times D_{\theta_b} \subset \C^2 \times \C$ containing  $( 0 , i \vartheta_0)$ for some fixed $\vartheta \in (0,\pi/2)$, where $S$  is  any bounded open subset of $\C^2$.
Specifically,  we will verify   Hypothesis \ref{HypI} - \ref{HypIII} needed for Theorem~2.10  in \cite{HasLan23-2}  using the following Table \ref{translationgesnondeg-2}.

\begin{table}[H] \caption{Notations}  \label{translationgesnondeg-2}
\begin{center}
\begin{tabular}{||c | c||}
 \hline
 \text{Notation here } & \text{Notation in  } \cite{HasLan23-2}  \\ [0.5ex]
 \hline\hline
 $(\underline{\kappa},\theta)$ & $s$  \\
 \hline
 $(0,i \vartheta_0)$ & $s_0$  \\
 \hline
 $\check{H}_{j;g}(\underline{\kappa},\theta )$   & $H_g(s) $ \\
 \hline
 $\check{H}_{{\rm el},j}(\theta )$   & $H_{\rm at}(s)$  \\
 \hline
 $\check{G}_{j; \underline{\kappa},\theta}$  & $G_{1,s}$  \\
 \hline
 $\check{G}_{j; \underline{\kappa},\theta}$  & $G_{2,s}$  \\
 \hline
 $ \check{W}_{j}(\underline{\kappa}, \theta) $ & $W(s) $ \\ [1ex]
 \hline
 $  \check{E}_{{\rm el},j}$ & $E_{\rm at}(s) $\\ [1ex]
 \hline
 $ P_{{\rm el},j}  $  & $P_{\rm at}(s)$  \\ [1ex]
 \hline
\end{tabular}
\end{center}
\end{table}

Once the Hypothesis  \ref{HypI} - \ref{HypIII}  are verified  the claim will follow from  Theorem~2.10  in \cite{HasLan23-2}.
Hypothesis  \ref{HypI} and  \ref{HypII}  are verified  as in the proof of Theorem \ref{gs:thmnondeg},  with the obvious modifications.

Let us now verify Hypothesis \ref{HypIII}. For this we will use Proposition \ref{veriatomIIIex}.
It follows from this proposition that  the following holds.
Let  $\vartheta_0 \in (0,\pi/2)$ and $\theta_1 \in (0,\infty)$.
Pick $\epsilon > 0$ such that $\epsilon < \vartheta_0 < \pi/2- \epsilon$.
Then for any $\rho \in (0,\delta_j / \check{\delta}_j )$ and
 $\mathcal{U}_0 :=    \{ (\theta , z ) \in \C^2  :  \epsilon  <   {\rm Im} \theta  <  \pi/2 - \epsilon  , \,  | {\rm Re} \theta |   < \theta_1  , |z| < \rho \sin(\epsilon) \} $ the following holds
$$
\sup_{(\theta,z)\in \UU}  \sup_{q \geq 0} \left\| \frac{q+1}{\check{H}_{{\rm el},j}(\theta) - e^{\theta} z + q} \overline{{P}}_{{\rm el},j}\right\|  < \infty .
$$
Now choose $\check{\delta}_j$ such that $\frac{1}{2} (\delta_j / \check{\delta}_j  ) \sin \epsilon   = 1$.
And so for $\rho = \frac{1}{2} (\delta_j / \check{\delta}_j  )$ we find
that  for $U_0 =\{  \theta  \in \C   :  \epsilon  <   {\rm Im} \theta  < \vartheta_b - \epsilon  , \,  | {\rm Re} \theta |   < \theta_1   \}$
\begin{align} \label{eq:frompropvariatomIIgs:6-2}
\sup_{\theta \in  U_0  ,  z \in D_1}  \sup_{q \geq 0} \left\| \frac{q+1}{\check{H}_{{\rm el},j}(\theta) - e^{\theta} z + q} \overline{{P}}_{{\rm el},j}\right\|  < \infty .
\end{align}
If
$
{\rm Re} \theta \geq - \ln 2  ,
$
then for any
$z \in D_{1/2}$ we have
\begin{align} \label{eq:frompropvariatomIIgs:7-2}
|e^{-\theta} z | = e^{-{\rm Re} \theta} |z| \leq  e^{-{\rm Re} \theta} \frac{1}{2} \leq 1
\end{align}
thus
$ e^{-\theta} D_{1/2} \subset D_1 $ and so  $D_{1/2} \subset e^{\theta} D_1 $. We conclude from \eqref{eq:frompropvariatomIIgs:6-2} and
 \eqref{eq:frompropvariatomIIgs:7-2}  that for
\begin{align*}
U =  U_0 \cap \{ \theta \in \C : {\rm Re} \theta \geq - \ln 2  \}
\end{align*}
we find
\begin{align} \label{IIIveri1}
\infty 	&  >  \sup_{(\theta,z) \in U \times D_{1/2} } \sup_{q \geq 0}
	\left\| \frac{q+1}{\check{H}_{{\rm el},j}(\theta)- z + q }
	\overline{P}_{{\rm el},j} \right\| \\
	&= 	 \sup_{(s,z) \in U  \times D_{1/2} } \sup_{q \geq 0}
	\left\| \frac{q+1}{H_{{\rm at}}(s)- z + q }
	\overline{P}_{{\rm at}}(s)\right\|  \nonumber
\end{align}
and hence
Hypothesis \ref{HypIII}  holds for the set $\mathcal{U} = U  \times D_{1/2}$.
Finally observe that by construction $D_{\theta_{\rm r}}(i\vartheta_0)  \subset U$ for some $\theta_{\rm r} >0$.

\end{proof}

\subsubsection{Proof of  Theorem \ref{gs:thmdegrot}}\label{proof:gs:thmdegrot}
\begin{proof}[Proof of  Theorem \ref{gs:thmdegrot} ]  The proof is almost identical  to the proof of Theorem  \ref{gs:thmnondeg}.
With  the difference that we in addition have to verify  the case of degeneracy in Hypothesis \ref{HypII}  (iii).
To this end we first recall the notation
introduced in Table  \ref{translationgesnondeg},  see also     \eqref{defhtildefirst} and \eqref{defwtildefirst}
and recall   \eqref{eq:Wdipinfields}.

By definition $  \RR(U) =   \RR_{\rm el}(U)  \otimes \RR_{\rm f}(U)  $ for $U \in SU(2)$.
Let $\mathcal{S} = \{  \RR(U) : U \in SU(2) \}$.
We will show the invariance of
 $\check{H}_{{\rm el},0}(\theta)$, see  \eqref{defhtildefirst},
  by considering the individual terms.  From  Hypothesis \ref{hyp:rsym} (i) we know that $V$ commutes with
$ \RR_{\rm el}(U)$.
It  is well known that $ \RR_{\rm el}(U)$ commutes with   the Laplacian, cf.  Lemma   \ref{lem:propoftfrot}.
The invariance of $H_f$ is also known  by Lemma   \ref{lem:propoftfrot}.

The  invariance of
\begin{align}
\check{W}(\underline{\kappa},\theta) & :=  e^{\theta}\check{\delta}_0^{-1}U_{\rm ph}(\tau)   W(\underline{\kappa},\theta) U_{\rm ph}(\tau)^{-1}   ,\label{defwtildefirst-2}
\end{align}
cf. \eqref{defwtildefirst}, i.e.    $  \RR(U)  \check{W}(\underline{\kappa},\theta) \RR(U)^*  =  \check{W}(\underline{\kappa},\theta)$
  is shown as follows.  From \eqref{defwtildefirst-2}
and the fact that $U_{\rm ph}$ commutes with $\mathcal{R}(U)$ we see that invariance of \eqref{defwtildefirst-2}
follows provided we show that
\begin{align}
\label{invofwtime0-2}   \RR(U) W(\underline{\kappa},\theta)  \RR(U)^*  = W(\underline{\kappa},\theta) \quad \text{ for all } U \in SU(2) .
\end{align}
To see \eqref{invofwtime0-2}  we recall that
\begin{align} \label{eq:Wdipinfields-1-2}
W(\underline{\kappa},\theta) = &   \sum_{j=1}^N \left( \kappa_1^3 \chi(  x_j)   x_j \cdot E_\theta(0) +\kappa_2^3 S_j \cdot B_\theta(0) \right)
\end{align}
Note that by unique analytic continuation if suffices to show invariance of  \eqref{eq:Wdipinfields-1-2}  for real $\theta$.
But this follows since in view of Lemma \ref{lem:propoftfrot} we have with $R: =\pi(U)^{-1}$
\begin{align*} 
 \RR(U) W(\underline{\kappa},\theta)  \RR(U)^*= &   \sum_{j=1}^N \big( \kappa_1^3 \chi( \RR(U) x_j \RR(U)^*)  \RR(U)   x_j  \RR(U)^* \cdot   \RR(U)  E_\theta(0) \RR(U)^* \\
 & \quad +  \kappa_2^5 \RR(U)   S_j \RR(U)^* \cdot   \RR(U)  B_\theta(0) \RR(U)^*\big) \\
 = &   \sum_{j=1}^N \big( \kappa_1^3 \chi(  R x_j )  R   x_j   \cdot  R  E_\theta(0)   + \kappa_2^5 R  S_j   \cdot   R B_\theta(0) \big)\\
 = &   \sum_{j=1}^N \big( \kappa_1^3 \chi(   x_j )     x_j   \cdot    E_\theta(0)   +  \kappa_2^5  S_j   \cdot    B_\theta(0) \big) = W(\underline{\kappa},\theta) .
\end{align*}
Thus we have shown \eqref{invofwtime0-2}, and so the invariance of $  \check{W}(\theta) $.

By Lemma~\ref{lem:propoftfrot} we know that   $ \RR_{\rm el}(U)$ is a unitary operator on $\HH_{\rm at} := \HH_{\rm el}$
and $ \RR_{\rm f}( U) $ is a unitary operator on $\FF = \FF(\hh)$.
The latter leaves the Fock vacuum as well as the one
particle sector invariant and commutes with dilation.
Furthermore, we have
$$
\mathcal{S}_1 = \{ S_1 : S_1 \otimes S_2 \in \mathcal{S} \} = \{  \RR_{\rm el}(U) : U \in SU(2) \} ,
$$
and this group acts by assumption  irreducibly on the eigenspace of $E_{{\rm el},0}$.
This finishes the verification of Hypothesis \ref{HypII} (iii). The rest of the proof is identical to the proof of Theorem~\ref{gs:thmnondeg}.

\end{proof}

\subsubsection{Proof of  Theorem \ref{gs:thmdegtime}} \label{proof:thmdegtime}
 \begin{proof}[Proof of  Theorem \ref{gs:thmdegtime} ]
 The proof is almost identical  to the proof of Theorem  \ref{gs:thmdegrot}.
With  the difference that we consider here time reversal symmetry.
We recall again the notation introduced in Table  \ref{translationgesnondeg},  see also     \eqref{defhtildefirst} and \eqref{defwtildefirst} and recall   \eqref{eq:Wdipinfields}.

Recall $ \mathcal{T} =  \mathcal{T}_{\rm el}  \otimes \mathcal{T}_{f}$ from \eqref{eq:defoftimerev}.
Let $\mathcal{S} = \{ \mathcal{T}   \}$.
To show  invariance of $\check{H}_{{\rm el},0}(\theta)$, see  \eqref{defhtildefirst},
  we consider the individual terms.
By Hypothesis \ref{hyp:tinvgs} we know that $$  \mathcal{T}_{\rm el} V_{\rm el} \mathcal{T}_{\rm el}^*  = V_{\rm el}^*. $$
Observe that by  Lemma \ref{lem:propoftf} we find
\begin{align*}
 \mathcal{T}_{\rm el} (- \Delta) \mathcal{T}_{\rm el}^* = - \Delta=   [ -  \Delta   ]^*
\end{align*}
and that   $H_f$ is symmetric with respect to $\mathcal{T}_f$. From the same Lemma we also get that $U_{\rm ph}(\tau)$ commutes with $\mathcal{T}_{f}$. Hence the invariance of
\begin{align}
  \check{W}(\underline{\kappa},\theta)  & :=  e^{\theta}\check{\delta}_0^{-1}U_{\rm ph}(\tau)   W(\underline{\kappa},\theta) U_{\rm ph}(\tau)^{-1}   ,
\end{align}
i.e., $\mathcal{T}\check{W}(\underline{\kappa},\theta)\mathcal{T}^* = \check{W}(\underline{\kappa},\theta)^*$ will follow by means  of $\mathcal{T}\mathcal{T}^* = 1$
provided
\begin{align} \label{eq:TWT=Wstar}
	\mathcal{T} W(\underline{\kappa},\theta) \mathcal{T}^* = W(\underline{\kappa},\theta)^*
\end{align}
and
\begin{align} \label{eq:TeT=estar}
	\mathcal{T} e^\theta \mathcal{T}^* = (e^\theta)^* .
\end{align}
By an analytic continuation argument it is enough to show \eqref{eq:TWT=Wstar} and \eqref{eq:TeT=estar} for real valued $\theta$.
Now  \eqref{eq:TeT=estar} follows directly from the definition.  To show \eqref{eq:TWT=Wstar} we recall that, see \eqref{eq:Wdipinfields},
\begin{align*} 
W(\underline{\kappa},\theta) = &   \sum_{j=1}^N \left( \kappa_1^3 \chi(  x_j)   x_j \cdot E_\theta(0) +  \kappa_2^5 S_j \cdot B_\theta(0) \right)  .
\end{align*}
From Lemma \ref{lem:propoftf} we see that for real $\theta \in \R$
\begin{align*} 
 \TT W(\underline{\kappa},\theta)  \TT^* = &   \sum_{j=1}^N \big( \overline{\kappa}_1^3 \chi( \TT x_j \TT^*)  \RR(U)   x_j  \TT^* \cdot   \TT  E_\theta(0) \TT^*+  \overline{\kappa}_2^5 \TT   S_j \TT^* \cdot   \TT  B_\theta(0) \TT^*\big) \\
 = &   \sum_{j=1}^N \big(\overline{\kappa}_1^3  \chi(   x_j )     x_j   \cdot    E_\theta(0)   + \overline{\kappa}_2^5  (-  S_j )   \cdot   (-  B_\theta(0) )  \big)\\
 = &   \sum_{j=1}^N \big( \overline{\kappa}_1^3 \chi(   x_j )     x_j   \cdot    E_\theta(0)   +  \overline{\kappa}_2^5   S_j   \cdot    B_\theta(0) \big) = [W(\underline{\kappa},\theta)]^* .
\end{align*}
Thus we have shown \eqref{eq:TWT=Wstar}, and so the invariance of $  \check{W}(\underline{\kappa},\theta) $.

By Lemma \ref{lem:propoftf} we also know that $\mathcal{T}_{\rm el}$ is an anti-unitary operator on $\HH_{\rm at} := \HH_{\rm el}$. Moreover, that $  \mathcal{T}_f  $ is an anti-unitary operator on $\FF = \FF(\hh)$.
The latter in addition leaves the Fock vacuum as well as the one particle sector invariant and commutes with dilation.
Furthermore, we have
$$
\mathcal{S}_1 = \{ S_1 : S_1 \otimes S_2 \in \mathcal{S} \} = \{ \mathcal{T}_{\rm el}  \} ,
$$
and this group acts by Lemma \ref{lem:timeirred}  irreducibly on the eigenspace of $E_{{\rm el},0}$.
\end{proof}

\subsection{Proof of Lemma \ref{Spinbahn}} \label{proof:Spinbahn}
\begin{proof}[Proof of Lemma \ref{Spinbahn} ]
This follows from the so called uncertainty principle lemma \cite[Lemma, Page 169]{ReeSim2}, which says that for all $\psi \in C_c^\infty(\R^3) $
\begin{equation} \label{eq:uncerprinc}
\int_{\R^3}  \frac{1}{4|x|^2} | \psi(x) |^2 dx \leq \int_{\R^3} | \nabla \psi(x) |^2 dx .
\end{equation}
Thus for any $\delta > 0$ and $\epsilon > 0$ and   $\psi \in C_c^\infty(\R^3) $
\begin{align}
& \| I_{SB} \psi \| \nonumber  \\
& \leq \sum_{j=1}^N \sum_{a,b=1}^3 2  \| \nu_{j,b}(x_j) \cdot   p_{j,a}  \psi \| \nonumber  \\
& \leq \sum_{j=1}^N \sum_{a,b=1}^3 2 \left(  \| 1_{|x_j| > \delta}  \nu_{j,b}(x_j)  p_{j,a}  \psi \| + \| 1_{|x_j| \leq \delta}  \nu_{j,b}(x_j)   p_{j,a}  \psi \| \right) \nonumber   \\
& \leq \sum_{j=1}^N \sum_{a,b=1}^3 2 \left(  \sup_{|x| > \delta} |   \nu_{j,b}(x)  | \|   p_{j,a}  \psi \| + 2 \sup_{|x|  \in[0,\delta]}  |  |x|\nu_{j,b}(x)|  \| (2|x_j|)^{-1}   p_{j,a}  \psi \| \right) \nonumber  \\
& \leq   2 \sum_{j=1}^N \sum_{b=1}^3 \sup_{|x| \geq \delta} |  \nu_{j,b}(x) | \sum_{a=1}^3 \| p_{j,a} \psi \|  +   4 \sum_{j=1}^N \sum_{b=1}^3 \sup_{|x|  \in[0,\delta]}  |  |x|\nu_{b,j}(x)| \sum_{a=1}^3 \|  \nabla p_{j,a}  \psi  \|  \label{spinbosrelbound-0} \\
& \leq   6 \sum_{j=1}^N \sum_{b=1}^3 \sup_{|x| \geq \delta} |  \nu_{j,b}(x) | \left( \frac{1}{4 \epsilon }  \| \psi \| +  \epsilon \| \Delta \psi \| \right) +   36 \sum_{j=1}^N \sum_{b=1}^3 \sup_{|x|  \in[0,\delta]}  |  |x|\nu_{b,j}(x)|   \|  \Delta  \psi  \| ,   \label{spinbosrelbound-1}
\end{align}
where in \eqref{spinbosrelbound-0} we used for the second term  \eqref{eq:uncerprinc} and the last line follows by multiple applications of Cauchy-Schwarz and the squaring inequality.
We see that we can make the expressions in the last line  \eqref{spinbosrelbound-1}  involving $\| \Delta \psi \|$  sufficiently small by choosing  first $\delta > 0$ sufficiently small, in view of   \eqref{UpperBoundNU}, and then $\epsilon > 0$.
This shows the statement about the infinitesimal boundedness, since $C_c^\infty(\R^3)$ is a core for $-\Delta$.  \\
Let us now show relative compactness.
To show compactness we assume $N=1$.
We insert $\psi = (p^2 + 1)^{-1} 1_{|p| \geq n } \phi$ into  the calculation leading to
the first   term in  \eqref{spinbosrelbound-0}  and second term in \eqref{spinbosrelbound-1}
\begin{align*}
& \| I_{SB} (p^2 + 1)^{-1} 1_{|p| \geq n } \phi \| \\
& \leq  2  \sum_{b=1}^3 \sup_{|x| \geq \delta} |  \nu_{1,b}(x) |
\|   p_{1,a}  (p^2 + 1)^{-1} 1_{|p| \geq n } \phi \| +
36 \sum_{b=1}^3 \sup_{|x|  \in[0,\delta]}  |  |x|\nu_{1,b}(x)|   \|  \Delta  (p^2 + 1)^{-1} 1_{|p| \geq n } \phi  \| \\
& \leq C_N n^{-1}  \sup_{r \geq \delta} |   \nu(r)  |  \| \phi \| +
C_N \sup_{r  \in[0,\delta]}  |  r\nu(r)|  \|\phi  \|.
\end{align*}
Now the right hand side can be made arbitrarily small by first choosing $\delta > 0$ sufficiently
small and then $n$ sufficiently large.
Thus  $I_{SB} (p^2 + 1)^{-1} 1_{|p| \leq n }$ converges in norm to  $I_{SB} (p^2 + 1)^{-1}$. So to show relative compactness
it suffices to show compactness of $I_{SB} (p^2 + 1)^{-1} 1_{|p| \leq n }$. But this follows since that operator is Hilbert Schmidt.
Let $\varphi_n$ be the Fourier transform of $|p| (p^2 + 1)^{-1} 1_{|p| \leq n }$. We can write for some constant $C$
\begin{align*}
   \| I_{SB} (p^2 + 1)^{-1} 1_{|p| \leq n } \|_2^2  & \leq   C  \int_{\R^3 \times \R^3}   \nu(|x|)^2   |\varphi_n(x-y)|^2   dy dx\\
   & \leq   C  \|\varphi_n\|^2 \int_{\R^3}   \nu(|x|)^2     dx  < \infty .
\end{align*}
This shows relative compactness.
\end{proof}


\appendix

\section{Hypothesis and Main Theorem from \cite{HasLan23-2}} \label{app:ResultsfromOtherPaper}
In the following we state, for the convenience of the reader, the Hypotheses and the Main Theorem from \cite{HasLan23-2}.

Let $X$ be an open subset of $\C^\nu$, where $\nu \in \N$.
For each $s \in X$ let  $H_{\rm at}(s)$ be a densely defined closed operator in some separable complex Hilbert space $\HH_{\rm at}$.
Let the Fock space $\FF$ be defined as in \eqref{eq:FockSpace} and the free field operator $H_{\rm f}$ as in \eqref{eq:fieldenergy}.
The interaction between atomic and field particles is given by
\begin{align*}
	W(s) :=   a(G_{1,\overline{s}}) +a^*(G_{2,s})  \,
\end{align*}
where $k\mapsto G_{j,s}(k)$ is an element of $L^2(\R^3\times \Z_2;\mathcal{L}(\HH_{\rm at}))$ for each $s \in X$.

\begin{hypr}\label{HypI}
For $s \in X$ and $j=1,2$ the mapping $s \mapsto G_{j,s}$ is a
bounded analytic function that has values in
$L^2(\R^3\times \Z_2;\mathcal{L}(\HH_{\rm at}))$.
Moreover there exists a $\mu >0$ such that
\begin{equation*}
	\max_{j=1,2} \sup_{s \in X} \|G_{j,s}\|_\mu < \infty \, .
\end{equation*}
\end{hypr}

\begin{hypr}\label{HypII}
\textrm{ }
\begin{itemize}
\item[(i)] The mapping $s \mapsto H_{\rm at}(s)$ is an analytic family in the sense of Kato.
\item[(ii)]  There exists $s_0 \in  X$  such that $E_{\rm at}(s_0)$ is a non-defective, discrete element of the spectrum of  $H_{\rm at}(s_0)$.
\item[(iii)] If $E_{\rm at}(s_0)$ is degenerate,  there
	exists a group of symmetries,
	$\mathcal{S}$, such that  $H_{\rm at}(s) \otimes \one_{\FF}$, $H_{\rm f}$,  and $W(s)$ are symmetric with respect to $\mathcal{S}$ for all   $s \in X$.
	Each element of  $\mathcal{S}$ can be written in the form
 $S_1 \otimes S_2$, where  $S_1$ is a symmetry in $\HH_{\rm at}$ and $S_2$ is a symmetry in $\FF$.
Furthermore, the set of symmetries in $\HH_{\rm at}$
$$
\mathcal{S}_1 :=  \{ S_1 :  S_1 \otimes S_2 \in \mathcal{S} \}
$$
acts irreducibly on the eigenspace of $H_{\rm at}(s_0)$ with eigenvalue $E_{\rm at}(s_0)$.
 Each element of $\mathcal{S}_2 := \{ S_2 :  S_1 \otimes S_2 \in \mathcal{S} \}$ leaves the
Fock vacuum as well as  the one particle subspace  invariant and commutes with the operator of dilation.
\end{itemize}
\end{hypr}

\begin{hypr}\label{HypIII}
Hypothesis~\ref{HypII} holds and there exists a neighborhood
$\mathcal{U} \subset X_1 \times \C$ of
$(s_0,E_{\rm at}(s_0))$
such that for all $(s,z) \in \mathcal{U}$ we have
$|E_{\rm at}(s) - z | < 1/2$,  $\sup_{(s,z) \in \mathcal{U} } \| P_{\rm at}(s) \| < \infty$,   and
\begin{equation*}
	\sup_{(s,z) \in \mathcal{U}} \sup_{q \geq 0}
	\left\| \frac{q+1}{H_{\rm at}(s)- z + q }
	\overline{P}_{\rm at}(s)\right\| < \infty \,.
\end{equation*}
\end{hypr}

\noindent
For a precise definition of $P_{\rm at}(s)$ and hence $\overline{P}_{\rm at}(s) := \one_{\HH_{\rm at}}-P_{\rm at}(s)$ we refer to \cite{HasLan23-2}.
For a subset $\Omega \subset \C^n$ we write  $\Omega^* := \{ \overline{z} :   z \in \Omega\}$.

\begin{hypr} \label{HypIV}
The following holds.
\begin{itemize}
\item[(i)]
We have  $X = X^*$ and  for all $s \in X$ the identities $G_{1,s} = G_{2,s}$ and $H_{\rm at}(s)^* = H_{\rm at}(\overline{s})$ hold.
 \item[(ii)]
We have $s_0 \in X \cap \R^\nu $ and $E_{\rm at}(s_0) = \inf \sigma (H_{\rm at }(s_0))$.
\end{itemize}
\end{hypr}

\begin{theorem} \label{thm:symdegenSpinBoson}
Suppose Hypotheses~\ref{HypI}, \ref{HypII}, \ref{HypIII} 
hold and  let $$d = \dim {\rm ker} ( H_{\rm at}(s_0) - E_{\rm at}(s_0) ) .$$
Then there exists a neighborhood $X_b \subset X$ of $s_0$
and a positive constant $g_b$ such that for all $s \in X_b$
and all $g \in [0, g_b]$
 the operator $H_g(s)$ has an eigenvalue $E_g(s)$ with  $\degendim$ linearly independent eigenvectors $\psi_{g,j}(s)$, $j=1,...,\degendim$,
with the following properties.
\begin{itemize}
\item[(i)] The functions $s \mapsto E_g(s)$  and $s \mapsto \psi_{g,j}(s)$  for $j=1,...,\degendim$  are analytic functions on $X_b$.
\item[(ii)]  Uniformly in $s \in X_b$
we have $\lim_{g \to 0} E_g(s) =  E_{\rm at}(s)$ and  $\lim_{g \to 0} \psi_{g,j}(s)  =  \varphi_{{\rm at},j}(s) \otimes  \Omega $ for some $\varphi_{{\rm at},j}(s) \in \ran P_{\rm at}(s)$.
\end{itemize}
If in addition Hypothesis~\ref{HypIV} holds, then   $X_{\rm b} = X_{\rm b}^*$  and
\begin{itemize}
\item[(iii)]  for all $s \in X_b \cap \R^\nu$ it holds that $E_g(s) = \inf \sigma(H_g(s)).$
\item[(iv)]  for all $s \in X_b $ it holds that $\overline{E}_g(s) =  E_g(\overline{s})$.
\end{itemize}
\end{theorem}

\section{Symmetries}\label{app:symmetries}
\begin{definition}
Let $\HH$ be a complex Hilbert space.
\begin{itemize}
\item[(a)]
A mapping $T : \HH \to \HH$ is called {\bf anti-linear} operator  in  $\HH$  if
$$
T(\alpha x + \beta y) = \overline{\alpha} T x + \overline{\beta}T y ,
$$
for all $\alpha, \beta \in \C$ and $x,y \in \HH$. An anti-linear $T$ operator is called bounded if $$\sup_{x: \| x \| \leq 1 } \| Tx \| < \infty . $$
\item[(b)]
The {\bf adjoint} of a bounded anti-linear operator,  $T : \HH \to \HH$, is defined to be
the anti-linear operator $T^* : \HH \to \HH$ such that
$$
\langle x , T y \rangle = \overline{ \langle T^* x , y \rangle }
$$
for all $x, y \in \HH$.
\item[(c)] An  anti-linear operator  $V$  in $\HH$  is called {\bf anti-unitary} if it is surjective and
$$
\langle V x , V y \rangle = \overline{ \langle  x , y \rangle }
$$
for all $x, y \in \HH$.
\end{itemize}
\end{definition}

\begin{definition}\label{def:symmetry}
Let $\HH$ be a complex Hilbert space.
\begin{itemize}
\item[(a)] A {\bf symmetry} in $\HH$ is a unitary or anti-unitary operator in $\HH$.
\item[(b)] We say that a  linear operator   $T$  in $\HH$ (possibly unbounded)  is  {\bf symmetric with respect to }  $S$   if
\begin{align*}
 S T S^* &= T  \,, \quad \textrm{for }  S
	\textrm{ unitary,} \\
 S T S^* &= T^* \,, \;\; \textrm{for } S
	\textrm{ antiunitary.}
\end{align*}
In that  we also say that $S$ is {\bf symmetry of}  $T$.
\end{itemize}
\end{definition}


\bibliography{references}

\providecommand{\bysame}{\leavevmode\hbox to3em{\hrulefill}\thinspace}
\providecommand{\MR}{\relax\ifhmode\unskip\space\fi MR }
\providecommand{\MRhref}[2]{%
  \href{http://www.ams.org/mathscinet-getitem?mr=#1}{#2}
}
\providecommand{\href}[2]{#2}
\begin{thebibliography}{10}

\bibitem{AbdHas12}
A.~Abdesselam and D.~Hasler, \emph{Analyticity of the ground state energy for
  massless {N}elson models}, Comm. Math. Phys. \textbf{310} (2012), no.~2,
  511--536. \MR{2890307}

\bibitem{BacBalPiz17}
Volker Bach, Miguel Ballesteros, and Alessandro Pizzo, \emph{Existence and
  construction of resonances for atoms coupled to the quantized radiation
  field}, Advances in Mathematics \textbf{314} (2017), 540--572.

\bibitem{BCFS}
Volker Bach, Thomas Chen, J{\"u}rg Fr{\"o}hlich, and Israel~Michael Sigal,
  \emph{Smooth {F}eshbach map and operator-theoretic renormalization group
  methods}, J. Funct. Anal. \textbf{203} (2003), no.~1, 44--92. \MR{1996868}

\bibitem{BacFroSig98-1}
Volker Bach, J{\"u}rg Fr{\"o}hlich, and Israel~Michael Sigal, \emph{Quantum
  electrodynamics of confined nonrelativistic particles}, Adv. Math.
  \textbf{137} (1998), no.~2, 299--395. \MR{1639713}

\bibitem{BacFroSig98-2}
\bysame, \emph{Renormalization group analysis of spectral problems in quantum
  field theory}, Adv. Math. \textbf{137} (1998), no.~2, 205--298. \MR{1639709}

\bibitem{BacFroSig99}
\bysame, \emph{Spectral analysis for systems of atoms and molecules coupled to
  the quantized radiation field}, Comm. Math. Phys. \textbf{207} (1999), no.~2,
  249--290. \MR{1724854}

\bibitem{BalDecHan.2019}
Miguel Ballesteros, Dirk-Andr\'{e} Deckert, and Felix H\"{a}nle,
  \emph{Analyticity of resonances and eigenvalues and spectral properties of
  the massless spin-boson model}, J. Funct. Anal. \textbf{276} (2019), no.~8,
  2524--2581. \MR{3926124}

\bibitem{BalFauFroSch15}
Miguel Ballesteros, J{\'e}r{\'e}my Faupin, J{\"u}rg Fr{\"o}hlich, and Baptiste
  Schubnel, \emph{Quantum electrodynamics of atomic resonances}, Comm. Math.
  Phys. \textbf{337} (2015), no.~2, 633--680. \MR{3339159}

\bibitem{CohDiuLal86v1}
C.~{Cohen-Tannoudji}, B.~{Diu}, and F.~{Laloe}, \emph{Quantum mechanics, volume
  1}, Wiley-VCH, June 1986.

\bibitem{CohDiuLal86v2}
\bysame, \emph{Quantum mechanics, volume 2}, Wiley-VCH, June 1986.

\bibitem{Die69}
J.~Dieudonn\'e, \emph{Foundations of modern analysis}, Pure and Applied
  Mathematics, vol. Vol. 10-I, Academic Press, New York-London, 1969, Enlarged
  and corrected printing. \MR{349288}

\bibitem{evans2010}
L.C. Evans, \emph{Partial differential equations}, Graduate studies in
  mathematics, American Mathematical Society, 2010.

\bibitem{Ger00}
Christian G{\'e}rard, \emph{On the existence of ground states for massless
  {P}auli-{F}ierz {H}amiltonians}, Ann. Henri Poincar\'e \textbf{1} (2000),
  no.~3, 443--459. \MR{1777307}

\bibitem{Gri04}
M.~Griesemer, \emph{Exponential decay and ionization thresholds in
  non-relativistic quantum electrodynamics}, Journal of Functional Analysis
  \textbf{210} (2004), no.~2, 321 -- 340.

\bibitem{GriHas09}
Marcel Griesemer and David Hasler, \emph{Analytic perturbation theory and
  renormalization analysis of matter coupled to quantized radiation}, Ann.
  Henri Poincar\'e \textbf{10} (2009), no.~3, 577--621. \MR{2519822}

\bibitem{GriLieLos01}
Marcel Griesemer, Elliott~H. Lieb, and Michael Loss, \emph{Ground states in
  non-relativistic quantum electrodynamics}, Invent. Math. \textbf{145} (2001),
  no.~3, 557--595. \MR{1856401}

\bibitem{HasHer11-2}
David Hasler and Ira Herbst, \emph{Convergent expansions in non-relativistic
  qed: analyticity of the ground state}, J. Funct. Anal. \textbf{261} (2011),
  no.~11, 3119--3154. \MR{2835993}

\bibitem{HasHer11-1}
\bysame, \emph{Ground states in the spin boson model}, Ann. Henri Poincar\'e
  \textbf{12} (2011), no.~4, 621--677. \MR{2787765}

\bibitem{HasLan18-1}
David Hasler and Markus Lange, \emph{Renormalization analysis for degenerate
  ground states}, J. Funct. Anal. \textbf{275} (2018), no.~1, 103--148.
  \MR{3799625}

\bibitem{HasLan23-2}
\bysame, \emph{Degenerate perturbation theory for models of quantum field
  theory with symmetries}, Annales Henri Poincaré (2023), 57.

\bibitem{HasLan23-1}
\bysame, \emph{Symmetries in non-relativistic quantum electrodynamics}, Reviews
  in Mathematical Physics \textbf{35} (2023), no.~08, 39.

\bibitem{Kra30}
Hendrik~A. Kramers, \emph{Th\'{e}orie g\'{e}n\'{e}rale de la rotation
  paramagn\'{e}tique dans les cristaux}, Proceedings Koninklijke Akademie van
  Wetenschappen \textbf{33} (1930), 959--972.

\bibitem{NIST}
Alexander Kramida and Yuri Ralchenko, \emph{Nist atomic spectra database, nist
  standard reference database 78}, 1999.

\bibitem{LieLos03}
Elliott~H. Lieb and Michael Loss, \emph{Existence of atoms and molecules in
  non-relativistic quantum electrodynamics}, Adv. Theor. Math. Phys. \textbf{7}
  (2003), no.~4, 667--710. \MR{2039034}

\bibitem{LosMiySpo09}
Michael Loss, Tadahiro Miyao, and Herbert Spohn, \emph{Kramers degeneracy
  theorem in nonrelativistic {QED}}, Lett. Math. Phys. \textbf{89} (2009),
  no.~1, 21--31. \MR{2520177}

\bibitem{Osgood1899}
W.~F. Osgood, \emph{Note {\"u}ber analytische functionen mehrerer
  ver{\"a}nderlichen}, Mathematische Annalen \textbf{52} (1899), no.~2–3,
  462–464.

\bibitem{ReeSim2}
M.~Reed and B.~Simon, \emph{Methods of modern mathematical physics. {II}.
  {F}ourier analysis, self-adjointness}, Academic Press, New York-London, 1975.
  \MR{0493420}

\bibitem{ReeSim4}
\bysame, \emph{Methods of {M}odern {M}athematical {P}hysics {IV}. {A}nalysis of
  {O}perators}, Academic Press, New York-London, 1978. \MR{0493421}

\bibitem{Schwabl2007}
Franz Schwabl, \emph{Quantum mechanics}, fourth ed., Springer, Berlin, 2007.
  \MR{2374994}

\bibitem{Sig09}
Israel~Michael Sigal, \emph{Ground state and resonances in the standard model
  of the non-relativistic {QED}}, J. Stat. Phys. \textbf{134} (2009), no.~5-6,
  899--939. \MR{2518974}

\bibitem{Sim.97}
Barry Simon, \emph{Representations of finite and compact groups}, Graduate
  Studies in Mathematics, vol.~10, American Mathematical Society, Providence,
  RI, 1996. \MR{1363490}

\bibitem{Spo98}
Herbert Spohn, \emph{Ground state of a quantum particle coupled to a scalar
  {B}ose field}, Lett. Math. Phys. \textbf{44} (1998), no.~1, 9--16.
  \MR{1623746}

\bibitem{Spo04}
\bysame, \emph{Dynamics of charged particles and their radiation field},
  Cambridge University Press, Cambridge, 2004. \MR{2097788}

\bibitem{Thaller1992}
Bernd Thaller, \emph{The dirac equation}, Springer Berlin Heidelberg, 1992.

\bibitem{Tha05}
\bysame, \emph{Advanced visual quantum mechanics}, Springer New York, 2005.

\bibitem{Zhi60}
Grigorii~M. Zhislin, \emph{Discussion of the spectrum of schr{\"o}dinger
  operators for systems of many particles}, Trudy Moskovskogo matematiceskogo
  obscestva \textbf{9} (1960), 81--120.

\end{thebibliography}
\bibliographystyle{amsplain}

\end{document}